\documentclass[pra,twocolumn,floatfix,superscriptaddress,notitlepage,longbibliography]{revtex4-1}

\usepackage[normalem]{ulem}
\usepackage{graphicx, color, graphpap}
\usepackage{enumitem}
\usepackage{amssymb}
\usepackage{amsthm}
\usepackage{multirow}
\usepackage[colorlinks=true,citecolor=blue,linkcolor=magenta]{hyperref}
\usepackage[T1]{fontenc}
\usepackage{verbatim}
\usepackage{mathtools}
\usepackage{titlesec}

\long\def\ca#1\cb{} 

\newcommand{\abs}[2][]{#1| #2 #1|}
\newcommand{\braket}[2]{\langle #1 \hspace{1pt} | \hspace{1pt} #2 \rangle}

\newcommand{\norm}[2][]{#1| \! #1| #2 #1| \! #1|}

\newcommand{\ket}[1]{|#1\rangle}               
\newcommand{\bra}[1]{\langle #1|}              
\newcommand{\dya}[1]{\ket{#1}\!\bra{#1}}
\newcommand{\matl}[3]{\langle #1|#2|#3\rangle} 

\newcommand{\HC}{\mathcal{H}}
\newcommand{\IC}{\mathcal{I}}

\newcommand{\OC}{\mathcal{O}}

\newcommand{\Tr}{{\rm Tr}}

\renewcommand{\geq}{\geqslant}
\renewcommand{\leq}{\leqslant}

\DeclareMathOperator*{\argmax}{arg\,max}

\newcommand{\ad}{^\dagger}

\newcommand*{\id}{\openone}

{}
{}
\newtheorem{theorem}{Theorem}
\newtheorem{lemma}{Lemma}
\newtheorem{corollary}{Corollary}
\newtheorem{proposition}{Proposition}

\newtheorem{definition}{Definition}

\begin{document}
\title{Sub-Quantum Fisher Information}

\author{M. Cerezo}
\thanks{The first two authors contributed equally to this work.}
\affiliation{Theoretical Division, Los Alamos National Laboratory, Los Alamos, New Mexico 87545, USA}
\affiliation{Center for Nonlinear Studies, Los Alamos National Laboratory, Los Alamos, New Mexico 87545, USA}
\affiliation{Quantum Science Center, Oak Ridge, TN 37931, USA}

\author{Akira Sone}
\thanks{The first two authors contributed equally to this work.}
\affiliation{Theoretical Division, Los Alamos National Laboratory, Los Alamos, New Mexico 87545, USA}
\affiliation{Center for Nonlinear Studies, Los Alamos National Laboratory, Los Alamos, New Mexico 87545, USA}
\affiliation{Quantum Science Center, Oak Ridge, TN 37931, USA}
\affiliation{Aliro Technologies, Inc. Boston, Massachusetts 02135, USA}

\author{Jacob L. Beckey}
\affiliation{Theoretical Division, Los Alamos National Laboratory, Los Alamos, New Mexico 87545, USA}
\affiliation{JILA, NIST and University of Colorado, Boulder, Colorado 80309}
\affiliation{Department of Physics, University of Colorado, Boulder, Colorado 80309}

\author{Patrick J. Coles}
\affiliation{Theoretical Division, Los Alamos National Laboratory, Los Alamos, New Mexico 87545, USA}
\affiliation{Quantum Science Center, Oak Ridge, TN 37931, USA}

\begin{abstract} 
The Quantum Fisher Information (QFI) plays a crucial role in quantum information theory and in many practical applications such as quantum metrology. However, computing the QFI is generally a computationally demanding task. In this work we analyze a lower bound on the QFI which we call the sub-Quantum Fisher Information (sub-QFI). The bound can be efficiently estimated on a quantum computer for an $n$-qubit state using $2n$ qubits. The sub-QFI is based on the super-fidelity, an upper bound on Uhlmann’s fidelity. We analyze the sub-QFI in the context of unitary families, where we derive several crucial properties including its geometrical interpretation. In particular, we prove that the QFI and the sub-QFI are maximized for the same optimal state, which implies that the sub-QFI is faithful to the QFI in the sense that both quantities share the same global extrema. Based on this faithfulness, the sub-QFI acts as an efficiently computable surrogate for the QFI for quantum sensing and quantum metrology applications. Finally, we provide additional meaning to the sub-QFI as a measure of coherence, asymmetry, and purity loss.
\end{abstract}
\maketitle

The Quantum Fisher Information (QFI) is a fundamental quantity in quantum information theory~\cite{petz2011introduction,hayashi2016quantum,pezze2009entanglement,modi2011quantum,Sone2018quantifying,sone2019nonclassical,takeoka2016optimal,schuff2020improving} with applications to quantum metrology~\cite{giovannetti2011advances,demkowicz2012elusive,degen2017quantum,pezze2018quantum,nichols2018multiparameter}, quantum channel estimation~\cite{katariya2020geometric}, and condensed matter physics~\cite{wang2014quantum,ye2016scaling,macieszczak2016dynamical}. 
Given a quantum state $\rho_{\theta}$ parameterized by a parameter $\theta$, the estimation uncertainty of $\theta$ is lower bounded by the inverse of the QFI $I(\rho_{\theta})$ via the quantum Cram\'{e}r-Rao bound. For instance, in a quantum metrology setting, the larger the QFI is, the more information about $\theta$ one can obtain from $\rho_{\theta}$. This indicates that the QFI can be regarded as a good measure to evaluate the sensing capabilities of a quantum state.

Despite its tremendous relevance, computing the QFI is usually a challenging task for both classical and quantum computers~\cite{watrous2002quantum}. Most classical methods that estimate $I(\rho_{\theta})$ require: 1) full knowledge of $\rho_{\theta}$, usually obtained via full quantum state tomography, and 2) diagonalization of density matrices~\cite{fiderer2020general}. However, both of these tasks are inefficient for large problem sizes because of the exponential scaling of the Hilbert space dimension. In addition, even if one can access $\rho_{\theta}$ on a quantum computer, no quantum algorithm can efficiently, directly compute the QFI due to its non-linear functional dependence on $\rho_{\theta}$ ~\cite{cerezo2020variationalfidelity}.

One can potentially circumvent the issue of exactly evaluating the QFI by instead  estimating efficiently computable bounds on $I(\rho_{\theta})$, as the latter provides a range in which the QFI actually lies~\cite{girolami2014observable,modi2016fragile,girolami2017detecting,toth2017lower,apellaniz2017optimal,sone2020generalized}. For metrology applications (where higher QFI is better), lower bounds are of particular interest as they can guarantee that a state $\rho_{\theta}$ has at least a certain QFI value. This can be used to check if the quality of the state for quantum sensing is above a desired threshold.

Recently, the advent of the so-called Noisy Intermediate Scale Quantum (NISQ) computers~\cite{preskill2018quantum} has opened  the exciting possibility of using these near-term devices for quantum sensing and quantum metrology tasks. Among the most promising applications using NISQ devices are variational quantum approaches~\cite{cerezo2020variationalreview,endo2021hybrid,bharti2021noisy,cerezo2020variational,bravo2020variational} where one attempts to variationally prepare the state that maximizes the QFI~\cite{koczor2020variational,yang2020probe,meyer2020variational,beckey2020variational,ma2020adaptive}. Indeed, most of these algorithms optimize lower bounds on the QFI by training a unitary that prepares the best possible input state.

In this work, we analyze a lower bound for the QFI, which we call the sub-Quantum Fisher Information (sub-QFI), and which can be efficiently computed on a quantum computer for an $n$ qubit state $\rho_\theta$ with $2n$ qubits. The properties of sub-QFI and its operational meaning are studied for states belonging to 
the unitary families. 

Our main results are as follows. First, we obtain a simple analytical expression for the sub-QFI for unitary families which allows us to analyze  its algebraic properties and its geometric interpretation. Second, we rigorously prove that the optimal mixed state that maximizes the QFI is always the same optimal state that maximizes the sub-QFI. This allows us to conclude that the sub-QFI is faithful to the QFI in the sense that both quantities share the same global extrema.  Then, we relate the sub-QFI to other relevant quantities in the literature to connect  the sub-QFI  with  coherence and asymmetry measures, and to the fragility of the probe state to stochastic error in metrology tasks.

Finally, we discuss using the sub-QFI as an efficiently computable surrogate for the QFI for quantum metrology applications in NISQ devices. Particularly, the fact that the sub-QFI can be efficiently computed with just $2n$ qubits makes it an attractive quantity. Moreover,  the faithfulness of the sub-QFI guarantees that if one maximizes this quantity, then the state prepared will also maximize the QFI. Hence, we expect that our contributions can be applied in the near-future for various practical metrology tasks using near-term quantum computers.

\bigskip

\textit{Sub-Quantum Fisher Information.}
Let $\rho$ be an $n$-qubit input state, known as the \textit{probe state}, which interacts with a source that encodes the information of a parameter $\theta$ of interest. The state resulting from the interaction, $\rho_\theta$, is known as the \textit{exact state}.  The ultimate precision $\Delta\theta$ in estimating $\theta$ from $\nu$ measurements on $\rho_\theta$ is bounded by the Quantum Fisher Information (QFI)  $I(\rho_\theta)$ via the quantum Cram\'er-Rao bound: 
\begin{align}
    (\Delta\theta)^2\geq \frac{1}{\nu I( \rho_\theta)}\,.
\label{eq:standardQFI2}
\end{align}
The QFI is defined as 
\begin{align}
    I(\rho_{\theta})=\Tr[L_{\theta}^2\rho_{\theta}]\,,
\label{eq:standardQFI}
\end{align}
where $L_{\theta}$ is the Symmetric Logarithmic Derivative (SLD) operator satisfying $\partial_{\theta}\rho_{\theta}=\frac{1}{2}(L_{\theta}\rho_{\theta}+\rho_{\theta}L_{\theta})$.
Moreover, defining the \textit{error state}, $\rho_{\theta+\delta}$, the QFI can also be related to Uhlmann's quantum fidelity from the equality~\cite{braunstein1994statistical,liu2014fidelity,Hayashi2004Quantum} 
\begin{align}
    I(\rho_{\theta})   & =-4\lim_{\delta\to 0}\partial_{\delta}^2 F(\rho_{\theta},\rho_{\theta+\delta})\label{eq:fidelity0}\\
    & =8\lim_{\delta\to 0}\frac{1-F(\rho_{\theta},\rho_{\theta+\delta})}{\delta^2}\,,\label{eq:fidelity}
\end{align}
where $F(\rho,\sigma)=|\!|\sqrt{\rho}\sqrt{\sigma}|\!|_{1}$ is the quantum fidelity, and  $\norm{A}_1=\Tr[\sqrt{AA\ad}]$ is the trace norm.  
Intuitively, Eq.~\eqref{eq:fidelity0} relates the QFI to how sensitive the probe state $\rho$ is to small parameter shifts, $\delta$. The greater the concavity of $F(\rho_{\theta},\rho_{\theta+\delta})$ (quantified by its second derivative with respect to $\delta$), the larger the QFI. Viewed another way, Eq.~\eqref{eq:fidelity} has the interpretation that the more distinguishable the exact and error states are, the more information one can obtain about the unknown parameter via measurement.

Based on Eq.~\eqref{eq:fidelity} the definition of the  sub-Quantum Fisher Information (sub-QFI) is as follows.
\begin{definition}\label{def:subQFI}
Let $\rho_{\theta}$ and $\rho_{\theta+\delta}$ be the exact and error state, respectively. The sub-QFI is defined as:
\begin{align}
\label{eq:subQFI}
    \mathcal{I}\left(\rho_{\theta}\right)=8\lim_{\delta\to 0}\frac{1-\sqrt{G(\rho_{\theta},\rho_{\theta+\delta})}}{\delta^2}\,,
\end{align}
where for two quantum states $\rho$ and $\sigma$
\begin{align}
\label{eq:superfidelity}
    G(\rho,\sigma)=\Tr\left[\rho\sigma\right]+\sqrt{(1-\Tr\left[\rho^2\right])(1-\Tr\left[\sigma^2\right])}\,,
\end{align}
is the so-called super-fidelity.
\end{definition}

As shown in~\cite{Miszczak2009sub}, the square root of the super-fidelity~\cite{mendoncca2008alternative,puchala2009bound} upper bounds Uhlmann's fidelity
\begin{align}\label{eq:supQFIbound}
F(\rho_{\theta},\rho_{\theta+\delta})\leq \sqrt{G(\rho_{\theta},\rho_{\theta+\delta})}\,.
\end{align}
Hence, combining~\eqref{eq:fidelity}, \eqref{eq:subQFI}, and~\eqref{eq:supQFIbound} we readily find that the sub-QFI is a lower bound for the QFI
\begin{align}\label{eq:subQFIbound}
  \mathcal{I}\left(\rho_{\theta}\right)\leq I\left(\rho_{\theta}\right)\,.
\end{align}

Equation~\eqref{eq:superfidelity} shows that each term in the sub-QFI can be efficiently estimated  as there exists an algorithm  that computes state overlaps of the form $\Tr[\rho_\theta\rho_{\theta+\delta}]$ with a depth-two quantum circuit and classical post-processing that scales linearly with $n$~\cite{cincio2018learning}. We note that the sub-QFI bound has been used before in the study of quantum coherence, entanglement witnesses, and in quantum optics~\cite{garttner2018relating, rivas2008intrinsic, rivas2010precision}. It is because of the  broad utility of the sub-QFI that an in-depth analysis of its properties is of such importance. 

\bigskip

\textit{Main results.} In what follows we consider unitary families, i.e., we analyze the case when the parameter $\theta$ is encoded in the probe state $\rho$ via
\begin{align}
    \rho_{\theta}=W_{\theta}\rho W_{\theta}\ad\, \quad \text{with} \quad W_{\theta}=e^{-i\theta H}\,,
\end{align}
for some $\theta$-independent Hermitian operator $H$, usually called the \textit{generator}. As shown in the Appendix, the following theorems and propositions hold.
\begin{theorem}[Explicit form]\label{th:unitary}
For unitary families, the sub-QFI of Definition~\ref{def:subQFI} can be expressed as 
\begin{align}
    \IC\left(\rho_{\theta}\right)&=4\left(\Tr\left[\rho^2H^2\right]-\Tr\left[\rho H \rho H\right]\right)\label{eq:Lquantity}\\
    &=-2\Tr\left[\left[\rho,H\right]^2\right] \label{eq:explicit-form-2}\,.
\end{align}
\end{theorem}
Note that Theorem~\ref{th:unitary} implies that $\IC\left(\rho_{\theta}\right)=\IC\left(\rho\right)$ since $[W_\theta,H]=0$, meaning that the sub-QFI is independent of $\theta$. However, we continue to use the notation $\IC\left(\rho\right)$ for consistency.  Given the spectral decomposition  $\rho=\sum_{j=1}^r\lambda_j\dya{\lambda_j}$, with $r=\text{rank}(\rho)$ and with $\lambda_i \geq \lambda_{i+1}$, we find from Eq.~\eqref{eq:Lquantity} that
\begin{equation}
    \IC\left(\rho\right)=2\sum_{i,j=1}^r(\lambda_i-\lambda_j)^2 |\matl{\lambda_i}{H}{\lambda_j}|^2\,.\label{eq:expression}
\end{equation}

The following properties of the sub-QFI can be readily derived from~\cite{luo2020skew}:
\begin{itemize}
    \item Non-negativity: $0\leq \IC\left(\rho\right)$.
    \item Convexity: $\IC\left(\sum_k \omega_k \rho^{(k)}\right)\leq \sum_k \omega_k \IC\left( \rho^{(k)}\right)$ for arbitrary quantum states $\rho^{(k)}$ and for positive coefficients $\omega_k$ such that $\sum_k \omega_k=1$.
    \item Modified additivity: Given a bipartite quantum system $\HC_A\otimes\HC_B$, let $\rho^A_\theta$ ($\rho_\theta^B$) and $H_A$ ($H_B$) respectively be a quantum state and a Hermitian operator on $\HC_A$ ($\HC_B$). Then, for $\rho_{\theta}=\rho_\theta^A\otimes\rho_\theta^B$ and $H=H_A\otimes \id+\id\otimes H_B$, we have $\IC(\rho_{\theta})=\Tr[(\rho_\theta^A)^2] \IC_A(\rho_\theta^{A})+\Tr[(\rho_\theta^B)^2] \IC_B(\rho_\theta^{B})$ where $\IC_{A(B)}(\rho_\theta^{A(B)})=-2\Tr\left[[\rho_\theta^{A(B)},H_{A(B)}]^2\right]$.
    \item Modified decreasing under the partial trace: Given a bipartite quantum system $\HC=\HC_A\otimes\HC_B$ and a state $\rho_\theta\in\HC$, let $\rho_\theta^A=\Tr_B[\rho_\theta]$ denote the reduced state on subsystem $A$. Then, for $H=H_A\otimes \id$ with $H_A$ a Hermitian operator in $\HC_A$, we have $\IC(\rho_{\theta})\geq\IC_A(\rho_\theta^{A})/d_B$, where $d_B$ is the dimension of $\HC_B$. 
\end{itemize}

Now, let us define an operator $\Lambda_{\theta}$ for a full rank probe state $\rho$,
\begin{align}
    \Lambda_{\theta} = \left(\partial_{\theta}\rho_{\theta}\right)\rho_{\theta}^{-1}\,, \quad 
    \Lambda_{\theta}\ad = \rho_{\theta}^{-1}\partial_{\theta}\rho_{\theta}\,,
\label{eq:nSLD}
\end{align}
which plays a role of a non-Hermitian SLD (nSLD) operator~\cite{alipour2015extended} since $     \partial_{\theta}\rho_{\theta} = \frac{1}{2}(\Lambda_{\theta}\rho_{\theta}+\rho_{\theta}\Lambda_{\theta}\ad)$
and $\Tr[\Lambda_{\theta}\rho_{\theta}]=\Tr[\rho_{\theta}\Lambda_{\theta}\ad]=0$.  Then, we can obtain the following result. 
\begin{proposition}[Non-Hermitian SLD operator]
\label{coro:SLD}
For full rank probe states $\rho$, the sub-QFI can be expressed as 
\begin{align}
    \IC\left(\rho_{\theta}\right)=2\Tr\left[\Lambda_{\theta}\ad\Lambda_{\theta}\rho_{\theta}^2\right]\,.
\end{align}
where $\Lambda_{\theta}$ is the  nSLD operator in~\eqref{eq:nSLD}. 
\end{proposition}

If $[\rho_{\theta},\partial_{\theta}\rho_{\theta}]=0$, then $[\rho_{\theta}^{-1},\partial_{\theta}\rho_{\theta}]=0$, $\Lambda_{\theta}$ becomes the standard SLD operator $L_\theta$ in Eq.~\eqref{eq:standardQFI}~\cite{paris2009quantum,liu2016quantum}. Note that when $\rho_{\theta}$ is a pure state $\dya{\psi_{\theta}}$, then Eq.~\eqref{eq:subQFIbound} is saturated because the super-fidelity becomes the standard fidelity when $\rho_{\theta}$ is a pure state. 
In general, however, the sub-QFI will only lower bound the QFI. Nonetheless, one can still derive the following theorem which shows that $\IC\left(\rho \right)$ and $I\left(\rho\right)$ are always simultaneously  maximized for the same optimal state.
\begin{theorem}[Optimal state]\label{th:max}
For any quantum state $\rho$ and for any Hermitian generator $H$, the sub-QFI of Definition~\ref{def:subQFI} and the  QFI in~\eqref{eq:fidelity} are maximized for the same state preparation $U \rho U\ad$, with the optimal unitary being
\small
\begin{align}\label{eq:Uopt}
    U_*=\argmax_{U} \IC\left(U\rho U\ad\right)=\argmax_{U}I\left(U\rho U\ad\right)\,.
\end{align}
\normalsize
Here, the maximum is taken over the unitary group of degree $d$, with $d=2^n$.
\end{theorem}
Theorem~\ref{th:max} shows that the sub-QFI and the QFI share the same optimal state $\rho^*=U_*\rho U_*\ad$. Hence, as shown in Fig.~\ref{fig:optimal}, sub-QFI and the QFI are simultaneously maximized by the input state $\rho^*$. Moreover, since $\IC\left(\rho_{\theta}\right)=0\leftrightarrow I\left(\rho_{\theta}\right)=0$ (as seen from~\eqref{eq:expression}),  then Theorem~\ref{th:max} implies the sub-QFI is \textit{faithful} to the QFI, in the sense that both QFI and sub-QFI share the same global extrema.

\begin{figure}
\centering
\includegraphics[width=0.95\columnwidth]{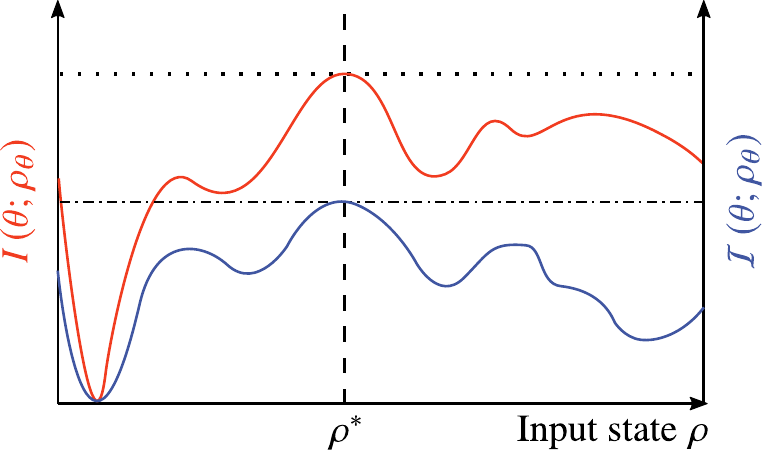}
\caption{\textbf{Schematic diagram of the QFI and sub-QFI versus the input state $\rho$.} The sub-QFI $\IC\left(\rho_{\theta}\right)$ is a lower bound for the QFI $I\left(\rho_{\theta}\right)$. As shown by Theorem~\ref{th:max}, the QFI and sub-QFI reach their maximum for the same optimal input state. Since $\IC\left(\rho_{\theta}\right)=0\leftrightarrow I\left(\rho_{\theta}\right)=0$, the sub-QFI is faithful to the QFI in the sense that both QFI and sub-QFI share the same global extrema. }
\label{fig:optimal}
\end{figure}

Let $\{h_j\}$ and $\{\ket{h_j}\}$ respectively be the set of eigenvalues and associated eigenvectors of $H$ (with $h_j\geq h_{j+1}$). Then, as shown in~\cite{fiderer2019maximal} the state that maximizes the QFI, and hence the sub-QFI, is given by $\rho^{*} =\sum_{j=1}^{d} \lambda_j \dya{\phi_j}$,  where
\begin{equation}\label{eq:optimal}
    \ket{\phi_j} =
    \begin{cases}
    \frac{1}{\sqrt{2}}\ket{h_j}+\frac{1}{\sqrt{2}}e^{i\chi}\ket{h_{d-j+1}}&(2j<d+1)\\
    \ket{h_j}&(2j=d+1)\\
   \frac{1}{\sqrt{2}}\ket{h_j}-\frac{1}{\sqrt{2}}e^{i\chi}\ket{h_{d-j+1}}&(2j>d+1)\\
    \end{cases}
\end{equation}
and where $\chi\in\mathbb{R}$ is an arbitrary phase.  Then, the following corollary holds.
\begin{corollary}[Maximal sub-QFI]
The maximal sub-QFI with respect to  all $d\times d$ state preparation unitaries is
\begin{align}
    \IC\left(\rho^*\right)&=\max_{U} \IC\left(U\rho U\ad\right)\\&
    =\frac{1}{2}\sum_{k=1}^d (\lambda_k-\lambda_{d-k+1})^2 \left(h_k-h_{d-k+1}\right)^2\,.
\end{align}
\end{corollary}

Finally, the following proposition provides a geometrical interpretation to the sub-QFI. 
\begin{proposition}[Geometrical interpretation] \label{prop:geo}
The sub-QFI can be expressed as
\begin{equation}\label{eq:DHS}
    \IC\left(\rho_{\theta}\right)= \partial^2_{\theta}D_{HS}(\rho,\rho_{\theta})\Bigr\rvert_{\theta = 0}\,,
\end{equation}
where $D_{HS}(A,B)=\Tr[(A-B)^2]$ is the Hilbert-Schmidt distance. 
\end{proposition}
Equation~\eqref{eq:DHS} shows that maximizing the sub-QFI $\IC\left(\rho_{\theta}\right)$ is equivalent to maximizing the susceptibility of the Hilbert-Schmidt distance between the probe and exact states to changes in $\theta$, and hence, minimizing the super-fidelity. Here we recall that an equation similar  to~\eqref{eq:DHS} can be obtained for the QFI by replacing the Hilbert-Schmidt distance by the Bures distance~\cite{hubner1992explicit,vsafranek2017discontinuities}, and for the Wigner-Yanase skew information in terms of the Hellinger distance~\cite{luo2004informational}. 

\bigskip

\textit{Connection to the literature.} Let us now discuss how the sub-QFI relates to other results in the literature, as this allows us to shed additional light on the properties of $\IC\left(\rho\right)$. First, we note that the quantity $-\frac{1}{2}\Tr\left[\left[\rho,H\right]^2\right]$ was introduced in~\cite{girolami2014observable}, as a lower-bound for the Wigner-Yanase skew information $I_{WY}(\rho,H)=-\frac{1}{2}\Tr\left[\left[\sqrt{\rho},H\right]^2\right]$~\cite{wigner1997information}. Hence, we find
\begin{align}\label{eq:subQFIbound2}
 I\left(\rho_{\theta}\right)\geq I_{WY}(\rho,H)\geq \frac{\mathcal{I}\left(\rho_{\theta}\right)}{8} \,.
\end{align}
Since the Wigner-Yanase skew information is  a measure of asymmetry~\cite{girolami2014observable,takagi2019skew,luo2017quantum,luo2003wigner} and a quasi-measure of the $H$-coherence~\cite{girolami2014observable,du2015wigner,marvian2016quantify} of the state $\rho$, then the sub-QFI can also be used as a bound for these quantities. We remark that in~\cite{yadin2016general}, the authors employ the quantity $-\frac{1}{2}\Tr\left[\left[\rho,H\right]^2\right]$  to derive a result similar to Proposition~\ref{prop:geo}.

\begin{figure}[t]
\centering
\includegraphics[width=1\columnwidth]{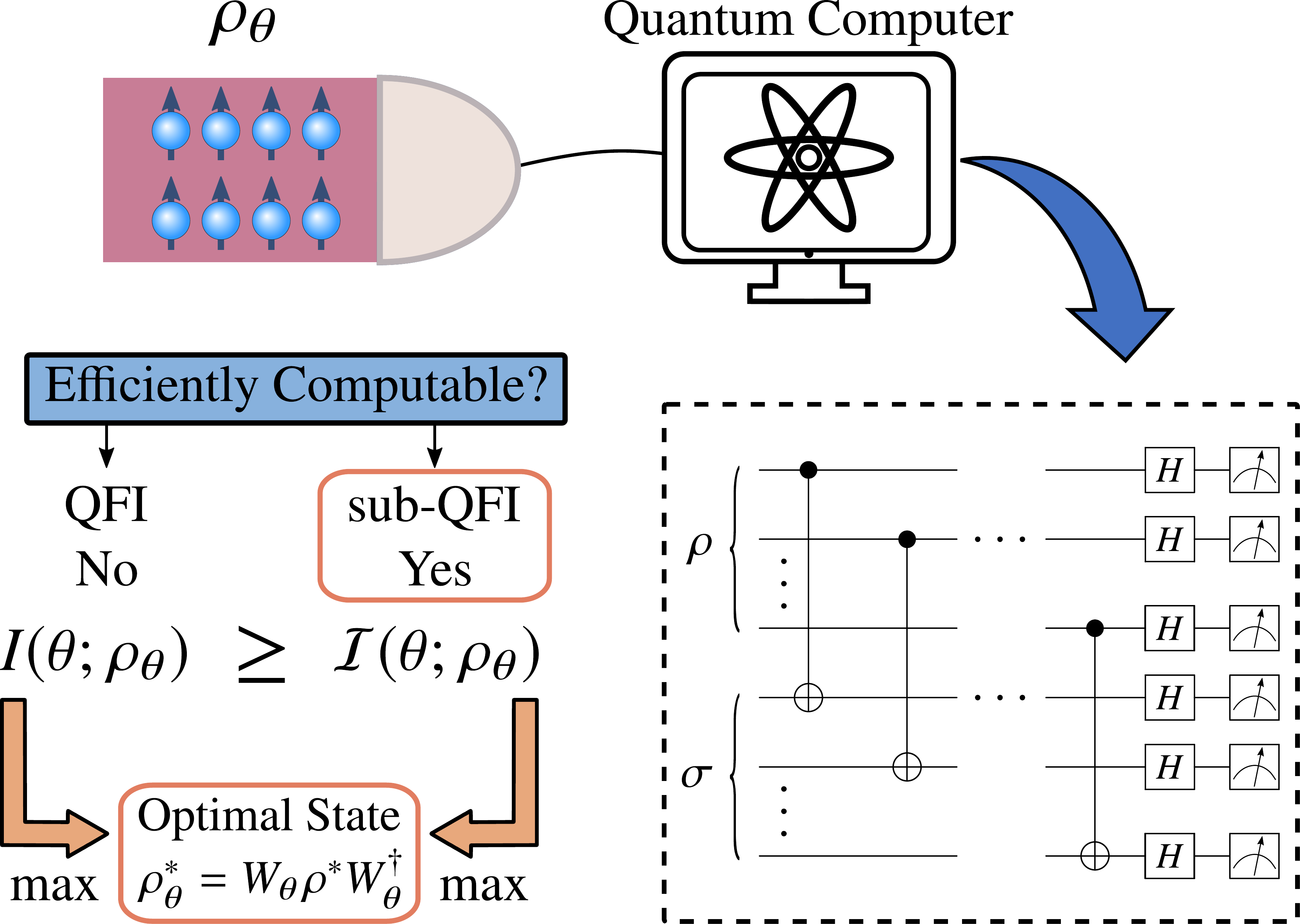}
\caption{\textbf{Use of sub-QFI for quantum metrology.} The sub-QFI can be efficiently computed from the exact state $\rho_\theta$ and error state $\rho_{\theta+\delta}$  via the depth-two circuit shown plus additional efficient classical post-processing~\cite{cincio2018learning} (here, $\rho=\rho_{\theta}$ and $\sigma=\rho_{\theta+\delta}$). Since the QFI and sub-QFI share the same optimal input state $\rho^*$, the sub-QFI can be used as an efficiently computable surrogate for the QFI when preparing states for quantum metrology.   }
\label{fig:fig2}
\end{figure}

The sub-QFI was also used in~\cite{garttner2018relating} to study of the relationship between quantum coherence and multiparticle entanglement~\cite{garttner2018relating}. Therein, the authors examine the relationship between multiple quantum coherence spectra (MQCs) and multiparticle entanglement witnesses. Specifically, it is shown that one can relate the MQCs to the QFI (a known entanglement witness) via the Sub-QFI in that the second moment of the MQC spectrum is equivalent to the Sub-QFI. Using the fact that the Sub-QFI inherits the property of being a witness of multi-particle entanglement from the QFI, they successfully establish the connection between MQCs and entanglement witnesses. 

In addition, the sub-QFI can be directly related to a quantity known as the purity loss~\cite{modi2016fragile,yang2020probe}, which indicates the sensing performance of the probe state $\rho$. Specifically, the purity loss $\Delta\gamma$ quantifies the robustness of $\rho$ to statistical fluctuations in the parameter $\theta$. For instance, given a random variable $X$ normally distributed about $\theta$ with variance $(\Delta x)^2$, and which takes value $x$ with probability $\lambda_x$, one can define the ensemble-averaged state $\rho_{\text{ave}}=\sum_x \lambda_x \rho_{\theta(x)}$  (with $\rho_{\theta(x)}=W_{\theta(x)}\rho W_{\theta(x)}\ad$).  Since the purity of $\rho_{\text{ave}}$ is smaller than that of $\rho$, the purity loss is defined as 
\begin{align}\label{eq:purityloss}
 \Delta\gamma = \Tr\left[\rho_{\theta}^2\right]-\Tr\left[\rho_{\text{ave}}^2\right]\,.
\end{align}
Here, $ \Delta\gamma$ quantifies the degree to which the probe state decoheres due to the statistical fluctuations in $\theta$~\cite{modi2016fragile}. Similarly, the ratio $\Delta\gamma/(\Delta x)^2$ characterizes the fragility of $\rho$ to this stochastic noise. 

From~\cite{modi2016fragile} we find that given a sharply distributed probability distribution of $X$ with variance $(\Delta x)^2$, the sub-QFI is related to the purity loss of~\eqref{eq:purityloss} via
\begin{equation}\label{eq:purlosssqfi}
    \IC\left(\rho_{\theta}\right)\approx 2 \frac{\Delta\gamma}{(\Delta x)^2}\,.
\end{equation}
Equation~\eqref{eq:purlosssqfi} shows that the sub-QFI quantifies how  fragile the probe state $\rho$ is to stochastic fluctuations in the parameter $\theta$.

Lastly, the sub-QFI has been explored in quantum optics applications \cite{rivas2008intrinsic,rivas2010precision,klimov2017optimal}. Specifically, in~\cite{rivas2008intrinsic} the authors derive the sub-QFI taking Eq.~\eqref{eq:DHS} in Proposition~\ref{prop:geo} as a starting point. Then, the authors explore the relevance of the sub-QFI in the context of quantum optics, showing that the sub-QFI can be used as a predictor of non-classicality in  states of light $\rho$ under certain conditions. After defining the quantity $A_H=H^2 -\text{:}H^2\text{:}$ (where $::$ denotes normal ordering of the operators), the  authors show that for classical states the QFI satisfies $I( \rho) \leq \Tr[\rho A_G]$. Hence, $I( \rho) > \Tr[\rho A_G]$ implies that $\rho $ is non-classical, as long as a linear scheme is used~\cite{rivas2010precision}. One can then readily conclude that the sub-QFI serves as an efficiently computable indicator of non-classicality as $\IC( \rho_{\theta}) > \Tr[\rho A_G]$ implies that $\rho $ is non-classical.

\textit{Sub-QFI for quantum metrology.} The fact that the sub-QFI can be efficiently estimated in a quantum computer and that it is faithful to the QFI means it can be employed to variationally prepare optimal states for quantum metrology.  Specifically, in this setting one trains a unitary $U$ acting on the input state with the goal of maximizing the sensing capabilities of $U\rho U\ad$~\cite{koczor2020variational,beckey2020variational}. As shown in Fig.~\ref{fig:fig2}, since one is unable to efficiently estimate the QFI, one instead computes the sub-QFI via the state overlap circuit~\cite{cincio2018learning}. Maximizing the sub-QFI means that one is able to guarantee a larger QFI value. While optimizing a lower bound could lead to a state being prepared that is not the optimal state for the QFI, Theorem~\ref{th:max} guarantees that the optimal state prepared by maximizing the sub-QFI will also maximize the QFI.

\textit{Conclusions.} We have presented and analyzed the sub-Quantum Fisher Information (sub-QFI), an efficiently computable lower bound of quantum Fisher information (QFI).  For unitary families, we have obtained rigorous results that allowed us to provide operational meaning the sub-QFI, including its connection to coherence and asymmetry measures, as well as its geometrical interpretation. In addition, we have proved that the sub-QFI is faithful to the QFI, in the sense that both quantities share the same global extrema. 

Our results pave the way for the sub-QFI to be employed to prepare the optimal state in quantum metrology and quantum sensing applications.
Moreover, due to the fact that the optimal state for sub-QFI can be deterministically prepared trough a variational circuit~\cite{beckey2020variational}, preparing the sub-QFI can overcome some of the limitations that have been to arise for probabilistic state preparation~\cite{hradil2019quantum}, and hence thus enhancing the utility of sub-QFI.

We further remark that due to its connection to many quantities of interest, the sub-QFI plays an important role in bridging  quantum metrology, quantum computing, and resource theories, and hence can help in advancing the understanding of the connection between these fields. Finally, we note that a future research direction is the study of the optimal measurement scheme associated with the optimal state for the sub-QFI.

\section*{Acknowledgements}

We thank Davide Girolami and Lukas J. Fiderer for helpful discussions. This work was supported by the Quantum Science Center (QSC), a National Quantum Information Science Research Center of the U.S. Department of Energy (DOE).  MC and AS also acknowledge initial support from the Center for Nonlinear Studies at Los Alamos National Laboratory (LANL). AS is now supported by the internal R\&D from Aliro Technologies, Inc.  JLB was initially supported by the U.S. DOE through a quantum computing program sponsored by the LANL Information Science \& Technology Institute. JLB was also supported by the National Science Foundation Graduate Research Fellowship under Grant No. 1650115.  PJC also acknowledges initial support from the LANL ASC Beyond Moore's Law project.

\bibliography{quantum.bib}

\begin{thebibliography}{64}%
\makeatletter
\providecommand \@ifxundefined [1]{%
 \@ifx{#1\undefined}
}%
\providecommand \@ifnum [1]{%
 \ifnum #1\expandafter \@firstoftwo
 \else \expandafter \@secondoftwo
 \fi
}%
\providecommand \@ifx [1]{%
 \ifx #1\expandafter \@firstoftwo
 \else \expandafter \@secondoftwo
 \fi
}%
\providecommand \natexlab [1]{#1}%
\providecommand \enquote  [1]{``#1''}%
\providecommand \bibnamefont  [1]{#1}%
\providecommand \bibfnamefont [1]{#1}%
\providecommand \citenamefont [1]{#1}%
\providecommand \href@noop [0]{\@secondoftwo}%
\providecommand \href [0]{\begingroup \@sanitize@url \@href}%
\providecommand \@href[1]{\@@startlink{#1}\@@href}%
\providecommand \@@href[1]{\endgroup#1\@@endlink}%
\providecommand \@sanitize@url [0]{\catcode `\\12\catcode `\$12\catcode
  `\&12\catcode `\#12\catcode `\^12\catcode `\_12\catcode `\%12\relax}%
\providecommand \@@startlink[1]{}%
\providecommand \@@endlink[0]{}%
\providecommand \url  [0]{\begingroup\@sanitize@url \@url }%
\providecommand \@url [1]{\endgroup\@href {#1}{\urlprefix }}%
\providecommand \urlprefix  [0]{URL }%
\providecommand \Eprint [0]{\href }%
\providecommand \doibase [0]{http://dx.doi.org/}%
\providecommand \selectlanguage [0]{\@gobble}%
\providecommand \bibinfo  [0]{\@secondoftwo}%
\providecommand \bibfield  [0]{\@secondoftwo}%
\providecommand \translation [1]{[#1]}%
\providecommand \BibitemOpen [0]{}%
\providecommand \bibitemStop [0]{}%
\providecommand \bibitemNoStop [0]{.\EOS\space}%
\providecommand \EOS [0]{\spacefactor3000\relax}%
\providecommand \BibitemShut  [1]{\csname bibitem#1\endcsname}%
\let\auto@bib@innerbib\@empty
\bibitem [{\citenamefont {Petz}\ and\ \citenamefont
  {Ghinea}(2011)}]{petz2011introduction}%
  \BibitemOpen
  \bibfield  {author} {\bibinfo {author} {\bibfnamefont {D{\'e}nes}\
  \bibnamefont {Petz}}\ and\ \bibinfo {author} {\bibfnamefont {Catalin}\
  \bibnamefont {Ghinea}},\ }\bibfield  {title} {\enquote {\bibinfo {title}
  {Introduction to quantum fisher information},}\ }in\ \href {\doibase
  10.1142/9789814338745_0015} {\emph {\bibinfo {booktitle} {Quantum probability
  and related topics}}}\ (\bibinfo  {publisher} {World Scientific},\ \bibinfo
  {year} {2011})\ pp.\ \bibinfo {pages} {261--281}\BibitemShut {NoStop}%
\bibitem [{\citenamefont {Hayashi}(2004{\natexlab{a}})}]{hayashi2016quantum}%
  \BibitemOpen
  \bibfield  {author} {\bibinfo {author} {\bibfnamefont {Masahito}\
  \bibnamefont {Hayashi}},\ }\href
  {https://www.springer.com/gp/book/9783662497234} {\emph {\bibinfo {title}
  {Quantum Information Theory: Mathematical Foundation (2nd edition)}}}\
  (\bibinfo  {publisher} {Springer},\ \bibinfo {year} {2004})\BibitemShut
  {NoStop}%
\bibitem [{\citenamefont {Pezz{\'{e}}}\ and\ \citenamefont
  {Smerzi}(2009)}]{pezze2009entanglement}%
  \BibitemOpen
  \bibfield  {author} {\bibinfo {author} {\bibfnamefont {Luca}\ \bibnamefont
  {Pezz{\'{e}}}}\ and\ \bibinfo {author} {\bibfnamefont {Augusto}\ \bibnamefont
  {Smerzi}},\ }\bibfield  {title} {\enquote {\bibinfo {title} {Entanglement,
  nonlinear dynamics, and the {H}eisenberg limit},}\ }\href {\doibase
  10.1103/PhysRevLett.102.100401} {\bibfield  {journal} {\bibinfo  {journal}
  {Phys. Rev. Lett.}\ }\textbf {\bibinfo {volume} {102}},\ \bibinfo {pages}
  {100401} (\bibinfo {year} {2009})}\BibitemShut {NoStop}%
\bibitem [{\citenamefont {Modi}\ \emph {et~al.}(2011)\citenamefont {Modi},
  \citenamefont {Cable}, \citenamefont {Williamson},\ and\ \citenamefont
  {Vedral}}]{modi2011quantum}%
  \BibitemOpen
  \bibfield  {author} {\bibinfo {author} {\bibfnamefont {Kavan}\ \bibnamefont
  {Modi}}, \bibinfo {author} {\bibfnamefont {Hugo}\ \bibnamefont {Cable}},
  \bibinfo {author} {\bibfnamefont {Mark}\ \bibnamefont {Williamson}}, \ and\
  \bibinfo {author} {\bibfnamefont {Vlatko}\ \bibnamefont {Vedral}},\
  }\bibfield  {title} {\enquote {\bibinfo {title} {Quantum correlations in
  mixed-state metrology},}\ }\href {\doibase 10.1103/PhysRevX.1.021022}
  {\bibfield  {journal} {\bibinfo  {journal} {Phys. Rev. X}\ }\textbf {\bibinfo
  {volume} {1}},\ \bibinfo {pages} {021022} (\bibinfo {year}
  {2011})}\BibitemShut {NoStop}%
\bibitem [{\citenamefont {Sone}\ \emph {et~al.}(2018)\citenamefont {Sone},
  \citenamefont {Zhuang},\ and\ \citenamefont
  {Cappellaro}}]{Sone2018quantifying}%
  \BibitemOpen
  \bibfield  {author} {\bibinfo {author} {\bibfnamefont {Akira}\ \bibnamefont
  {Sone}}, \bibinfo {author} {\bibfnamefont {Quntao}\ \bibnamefont {Zhuang}}, \
  and\ \bibinfo {author} {\bibfnamefont {Paola}\ \bibnamefont {Cappellaro}},\
  }\bibfield  {title} {\enquote {\bibinfo {title} {Quantifying precision loss
  in local quantum thermometry via diagonal discord},}\ }\href {\doibase
  10.1103/PhysRevA.98.012115} {\bibfield  {journal} {\bibinfo  {journal} {Phys.
  Rev. A}\ }\textbf {\bibinfo {volume} {98}},\ \bibinfo {pages} {012115}
  (\bibinfo {year} {2018})}\BibitemShut {NoStop}%
\bibitem [{\citenamefont {Sone}\ \emph {et~al.}(2019)\citenamefont {Sone},
  \citenamefont {Zhuang}, \citenamefont {Li}, \citenamefont {Liu},\ and\
  \citenamefont {Cappellaro}}]{sone2019nonclassical}%
  \BibitemOpen
  \bibfield  {author} {\bibinfo {author} {\bibfnamefont {Akira}\ \bibnamefont
  {Sone}}, \bibinfo {author} {\bibfnamefont {Quntao}\ \bibnamefont {Zhuang}},
  \bibinfo {author} {\bibfnamefont {Changhao}\ \bibnamefont {Li}}, \bibinfo
  {author} {\bibfnamefont {Yi-Xiang}\ \bibnamefont {Liu}}, \ and\ \bibinfo
  {author} {\bibfnamefont {Paola}\ \bibnamefont {Cappellaro}},\ }\bibfield
  {title} {\enquote {\bibinfo {title} {Nonclassical correlations for quantum
  metrology in thermal equilibrium},}\ }\href {\doibase
  10.1103/PhysRevA.99.052318} {\bibfield  {journal} {\bibinfo  {journal} {Phys.
  Rev. A}\ }\textbf {\bibinfo {volume} {99}},\ \bibinfo {pages} {052318}
  (\bibinfo {year} {2019})}\BibitemShut {NoStop}%
\bibitem [{\citenamefont {Takeoka}\ and\ \citenamefont
  {Wilde}(2016)}]{takeoka2016optimal}%
  \BibitemOpen
  \bibfield  {author} {\bibinfo {author} {\bibfnamefont {Masahiro}\
  \bibnamefont {Takeoka}}\ and\ \bibinfo {author} {\bibfnamefont {Mark~M}\
  \bibnamefont {Wilde}},\ }\bibfield  {title} {\enquote {\bibinfo {title}
  {Optimal estimation and discrimination of excess noise in thermal and
  amplifier channels},}\ }\href {https://arxiv.org/abs/1611.09165} {\bibfield
  {journal} {\bibinfo  {journal} {arXiv preprint arXiv:1611.09165}\ } (\bibinfo
  {year} {2016})}\BibitemShut {NoStop}%
\bibitem [{\citenamefont {Schuff}\ \emph {et~al.}(2020)\citenamefont {Schuff},
  \citenamefont {Fiderer},\ and\ \citenamefont {Braun}}]{schuff2020improving}%
  \BibitemOpen
  \bibfield  {author} {\bibinfo {author} {\bibfnamefont {Jonas}\ \bibnamefont
  {Schuff}}, \bibinfo {author} {\bibfnamefont {Lukas~J}\ \bibnamefont
  {Fiderer}}, \ and\ \bibinfo {author} {\bibfnamefont {Daniel}\ \bibnamefont
  {Braun}},\ }\bibfield  {title} {\enquote {\bibinfo {title} {Improving the
  dynamics of quantum sensors with reinforcement learning},}\ }\href {\doibase
  10.1088/1367-2630/ab6f1f} {\bibfield  {journal} {\bibinfo  {journal} {New
  Journal of Physics}\ }\textbf {\bibinfo {volume} {22}},\ \bibinfo {pages}
  {035001} (\bibinfo {year} {2020})}\BibitemShut {NoStop}%
\bibitem [{\citenamefont {Giovannetti}\ \emph {et~al.}(2011)\citenamefont
  {Giovannetti}, \citenamefont {Lloyd},\ and\ \citenamefont
  {Maccone}}]{giovannetti2011advances}%
  \BibitemOpen
  \bibfield  {author} {\bibinfo {author} {\bibfnamefont {Vittorio}\
  \bibnamefont {Giovannetti}}, \bibinfo {author} {\bibfnamefont {Seth}\
  \bibnamefont {Lloyd}}, \ and\ \bibinfo {author} {\bibfnamefont {Lorenzo}\
  \bibnamefont {Maccone}},\ }\bibfield  {title} {\enquote {\bibinfo {title}
  {Advances in quantum metrology},}\ }\href
  {https://www.nature.com/articles/nphoton.2011.35} {\bibfield  {journal}
  {\bibinfo  {journal} {Nat. Photonics}\ }\textbf {\bibinfo {volume} {5}},\
  \bibinfo {pages} {222--229} (\bibinfo {year} {2011})}\BibitemShut {NoStop}%
\bibitem [{\citenamefont {Demkowicz-Dobrza{\'n}ski}\ \emph
  {et~al.}(2012)\citenamefont {Demkowicz-Dobrza{\'n}ski}, \citenamefont
  {Ko{\l}ody{\'n}ski},\ and\ \citenamefont
  {Gu{\c{t}}{\u{a}}}}]{demkowicz2012elusive}%
  \BibitemOpen
  \bibfield  {author} {\bibinfo {author} {\bibfnamefont {Rafa{\l}}\
  \bibnamefont {Demkowicz-Dobrza{\'n}ski}}, \bibinfo {author} {\bibfnamefont
  {Jan}\ \bibnamefont {Ko{\l}ody{\'n}ski}}, \ and\ \bibinfo {author}
  {\bibfnamefont {M{\u{a}}d{\u{a}}lin}\ \bibnamefont {Gu{\c{t}}{\u{a}}}},\
  }\bibfield  {title} {\enquote {\bibinfo {title} {The elusive {H}eisenberg
  limit in quantum-enhanced metrology},}\ }\href {\doibase 10.1038/ncomms2067}
  {\bibfield  {journal} {\bibinfo  {journal} {Nature communications}\ }\textbf
  {\bibinfo {volume} {3}},\ \bibinfo {pages} {1--8} (\bibinfo {year}
  {2012})}\BibitemShut {NoStop}%
\bibitem [{\citenamefont {Degen}\ \emph {et~al.}(2017)\citenamefont {Degen},
  \citenamefont {Reinhard},\ and\ \citenamefont
  {Cappellaro}}]{degen2017quantum}%
  \BibitemOpen
  \bibfield  {author} {\bibinfo {author} {\bibfnamefont {C.~L.}\ \bibnamefont
  {Degen}}, \bibinfo {author} {\bibfnamefont {F.}~\bibnamefont {Reinhard}}, \
  and\ \bibinfo {author} {\bibfnamefont {P.}~\bibnamefont {Cappellaro}},\
  }\bibfield  {title} {\enquote {\bibinfo {title} {Quantum sensing},}\ }\href
  {\doibase 10.1103/RevModPhys.89.035002} {\bibfield  {journal} {\bibinfo
  {journal} {Rev. Mod. Phys.}\ }\textbf {\bibinfo {volume} {89}},\ \bibinfo
  {pages} {035002} (\bibinfo {year} {2017})}\BibitemShut {NoStop}%
\bibitem [{\citenamefont {Pezz\`e}\ \emph {et~al.}(2018)\citenamefont
  {Pezz\`e}, \citenamefont {Smerzi}, \citenamefont {Oberthaler}, \citenamefont
  {Schmied},\ and\ \citenamefont {Treutlein}}]{pezze2018quantum}%
  \BibitemOpen
  \bibfield  {author} {\bibinfo {author} {\bibfnamefont {Luca}\ \bibnamefont
  {Pezz\`e}}, \bibinfo {author} {\bibfnamefont {Augusto}\ \bibnamefont
  {Smerzi}}, \bibinfo {author} {\bibfnamefont {Markus~K.}\ \bibnamefont
  {Oberthaler}}, \bibinfo {author} {\bibfnamefont {Roman}\ \bibnamefont
  {Schmied}}, \ and\ \bibinfo {author} {\bibfnamefont {Philipp}\ \bibnamefont
  {Treutlein}},\ }\bibfield  {title} {\enquote {\bibinfo {title} {Quantum
  metrology with nonclassical states of atomic ensembles},}\ }\href {\doibase
  10.1103/RevModPhys.90.035005} {\bibfield  {journal} {\bibinfo  {journal}
  {Rev. Mod. Phys.}\ }\textbf {\bibinfo {volume} {90}},\ \bibinfo {pages}
  {035005} (\bibinfo {year} {2018})}\BibitemShut {NoStop}%
\bibitem [{\citenamefont {Nichols}\ \emph {et~al.}(2018)\citenamefont
  {Nichols}, \citenamefont {Liuzzo-Scorpo}, \citenamefont {Knott},\ and\
  \citenamefont {Adesso}}]{nichols2018multiparameter}%
  \BibitemOpen
  \bibfield  {author} {\bibinfo {author} {\bibfnamefont {Rosanna}\ \bibnamefont
  {Nichols}}, \bibinfo {author} {\bibfnamefont {Pietro}\ \bibnamefont
  {Liuzzo-Scorpo}}, \bibinfo {author} {\bibfnamefont {Paul~A}\ \bibnamefont
  {Knott}}, \ and\ \bibinfo {author} {\bibfnamefont {Gerardo}\ \bibnamefont
  {Adesso}},\ }\bibfield  {title} {\enquote {\bibinfo {title} {Multiparameter
  gaussian quantum metrology},}\ }\href {\doibase 10.1103/PhysRevA.98.012114}
  {\bibfield  {journal} {\bibinfo  {journal} {Physical Review A}\ }\textbf
  {\bibinfo {volume} {98}},\ \bibinfo {pages} {012114} (\bibinfo {year}
  {2018})}\BibitemShut {NoStop}%
\bibitem [{\citenamefont {Katariya}\ and\ \citenamefont
  {Wilde}(2020)}]{katariya2020geometric}%
  \BibitemOpen
  \bibfield  {author} {\bibinfo {author} {\bibfnamefont {Vishal}\ \bibnamefont
  {Katariya}}\ and\ \bibinfo {author} {\bibfnamefont {Mark~M}\ \bibnamefont
  {Wilde}},\ }\bibfield  {title} {\enquote {\bibinfo {title} {Geometric
  distinguishability measures limit quantum channel estimation and
  discrimination},}\ }\href {https://arxiv.org/abs/2004.10708} {\bibfield
  {journal} {\bibinfo  {journal} {arXiv preprint arXiv:2004.10708}\ } (\bibinfo
  {year} {2020})}\BibitemShut {NoStop}%
\bibitem [{\citenamefont {Wang}\ \emph {et~al.}(2014)\citenamefont {Wang},
  \citenamefont {Wu}, \citenamefont {Yang}, \citenamefont {Jin}, \citenamefont
  {Lambert},\ and\ \citenamefont {Nori}}]{wang2014quantum}%
  \BibitemOpen
  \bibfield  {author} {\bibinfo {author} {\bibfnamefont {Teng-Long}\
  \bibnamefont {Wang}}, \bibinfo {author} {\bibfnamefont {Ling-Na}\
  \bibnamefont {Wu}}, \bibinfo {author} {\bibfnamefont {Wen}\ \bibnamefont
  {Yang}}, \bibinfo {author} {\bibfnamefont {Guang-Ri}\ \bibnamefont {Jin}},
  \bibinfo {author} {\bibfnamefont {Neill}\ \bibnamefont {Lambert}}, \ and\
  \bibinfo {author} {\bibfnamefont {Franco}\ \bibnamefont {Nori}},\ }\bibfield
  {title} {\enquote {\bibinfo {title} {Quantum fisher information as a
  signature of the superradiant quantum phase transition},}\ }\href {\doibase
  10.1088/1367-2630/16/6/063039} {\bibfield  {journal} {\bibinfo  {journal}
  {New J. Phys.}\ }\textbf {\bibinfo {volume} {16}},\ \bibinfo {pages} {063039}
  (\bibinfo {year} {2014})}\BibitemShut {NoStop}%
\bibitem [{\citenamefont {Ye}\ \emph {et~al.}(2016)\citenamefont {Ye},
  \citenamefont {Hu},\ and\ \citenamefont {Wu}}]{ye2016scaling}%
  \BibitemOpen
  \bibfield  {author} {\bibinfo {author} {\bibfnamefont {En-Jia}\ \bibnamefont
  {Ye}}, \bibinfo {author} {\bibfnamefont {Zheng-Da}\ \bibnamefont {Hu}}, \
  and\ \bibinfo {author} {\bibfnamefont {Wei}\ \bibnamefont {Wu}},\ }\bibfield
  {title} {\enquote {\bibinfo {title} {Scaling of quantum fisher information
  close to the quantum phase transition in the xy spin chain},}\ }\href
  {\doibase
  https://www.sciencedirect.com/science/article/abs/pii/S0921452616303891}
  {\bibfield  {journal} {\bibinfo  {journal} {Physica B: Condensed Matter}\
  }\textbf {\bibinfo {volume} {502}},\ \bibinfo {pages} {151} (\bibinfo {year}
  {2016})}\BibitemShut {NoStop}%
\bibitem [{\citenamefont {Macieszczak}\ \emph {et~al.}(2016)\citenamefont
  {Macieszczak}, \citenamefont {Gu{\c{t}}{\u{a}}}, \citenamefont {Lesanovsky},\
  and\ \citenamefont {Garrahan}}]{macieszczak2016dynamical}%
  \BibitemOpen
  \bibfield  {author} {\bibinfo {author} {\bibfnamefont {Katarzyna}\
  \bibnamefont {Macieszczak}}, \bibinfo {author} {\bibfnamefont
  {M{\u{a}}d{\u{a}}lin}\ \bibnamefont {Gu{\c{t}}{\u{a}}}}, \bibinfo {author}
  {\bibfnamefont {Igor}\ \bibnamefont {Lesanovsky}}, \ and\ \bibinfo {author}
  {\bibfnamefont {Juan~P}\ \bibnamefont {Garrahan}},\ }\bibfield  {title}
  {\enquote {\bibinfo {title} {Dynamical phase transitions as a resource for
  quantum enhanced metrology},}\ }\href {\doibase 10.1103/PhysRevA.93.022103}
  {\bibfield  {journal} {\bibinfo  {journal} {Physical Review A}\ }\textbf
  {\bibinfo {volume} {93}},\ \bibinfo {pages} {022103} (\bibinfo {year}
  {2016})}\BibitemShut {NoStop}%
\bibitem [{\citenamefont {Watrous}(0202111)}]{watrous2002quantum}%
  \BibitemOpen
  \bibfield  {author} {\bibinfo {author} {\bibfnamefont {Joh}\ \bibnamefont
  {Watrous}},\ }\bibfield  {title} {\enquote {\bibinfo {title} {Quantum
  statistical zero-knowledge},}\ }\href
  {https://arxiv.org/abs/quant-ph/0202111} {\bibfield  {journal} {\bibinfo
  {journal} {arXiv preprint arXiv:0202111}\ } (\bibinfo {year}
  {0202111})}\BibitemShut {NoStop}%
\bibitem [{\citenamefont {Fiderer}\ \emph {et~al.}(2020)\citenamefont
  {Fiderer}, \citenamefont {Tufarelli}, \citenamefont {Piano},\ and\
  \citenamefont {Adesso}}]{fiderer2020general}%
  \BibitemOpen
  \bibfield  {author} {\bibinfo {author} {\bibfnamefont {Lukas~J}\ \bibnamefont
  {Fiderer}}, \bibinfo {author} {\bibfnamefont {Tommaso}\ \bibnamefont
  {Tufarelli}}, \bibinfo {author} {\bibfnamefont {Samanta}\ \bibnamefont
  {Piano}}, \ and\ \bibinfo {author} {\bibfnamefont {Gerardo}\ \bibnamefont
  {Adesso}},\ }\bibfield  {title} {\enquote {\bibinfo {title} {General
  expressions for the quantum fisher information matrix with applications to
  discrete quantum imaging},}\ }\href {https://arxiv.org/abs/2012.01572}
  {\bibfield  {journal} {\bibinfo  {journal} {arXiv preprint arXiv:2012.01572}\
  } (\bibinfo {year} {2020})}\BibitemShut {NoStop}%
\bibitem [{\citenamefont {Cerezo}\ \emph
  {et~al.}(2020{\natexlab{a}})\citenamefont {Cerezo}, \citenamefont {Poremba},
  \citenamefont {Cincio},\ and\ \citenamefont
  {Coles}}]{cerezo2020variationalfidelity}%
  \BibitemOpen
  \bibfield  {author} {\bibinfo {author} {\bibfnamefont {M.}~\bibnamefont
  {Cerezo}}, \bibinfo {author} {\bibfnamefont {Alexander}\ \bibnamefont
  {Poremba}}, \bibinfo {author} {\bibfnamefont {Lukasz}\ \bibnamefont
  {Cincio}}, \ and\ \bibinfo {author} {\bibfnamefont {Patrick~J}\ \bibnamefont
  {Coles}},\ }\bibfield  {title} {\enquote {\bibinfo {title} {Variational
  quantum fidelity estimation},}\ }\href {\doibase 10.22331/q-2020-03-26-248}
  {\bibfield  {journal} {\bibinfo  {journal} {Quantum}\ }\textbf {\bibinfo
  {volume} {4}},\ \bibinfo {pages} {248} (\bibinfo {year}
  {2020}{\natexlab{a}})}\BibitemShut {NoStop}%
\bibitem [{\citenamefont {Girolami}(2014)}]{girolami2014observable}%
  \BibitemOpen
  \bibfield  {author} {\bibinfo {author} {\bibfnamefont {Davide}\ \bibnamefont
  {Girolami}},\ }\bibfield  {title} {\enquote {\bibinfo {title} {Observable
  measure of quantum coherence in finite dimensional systems},}\ }\href
  {\doibase 10.1103/PhysRevLett.113.170401} {\bibfield  {journal} {\bibinfo
  {journal} {Phys. Rev. Lett.}\ }\textbf {\bibinfo {volume} {113}},\ \bibinfo
  {pages} {170401} (\bibinfo {year} {2014})}\BibitemShut {NoStop}%
\bibitem [{\citenamefont {Modi}\ \emph {et~al.}(2016)\citenamefont {Modi},
  \citenamefont {C{\'{e}}leri}, \citenamefont {Thompson},\ and\ \citenamefont
  {Gu}}]{modi2016fragile}%
  \BibitemOpen
  \bibfield  {author} {\bibinfo {author} {\bibfnamefont {Kavan}\ \bibnamefont
  {Modi}}, \bibinfo {author} {\bibfnamefont {Lucas~C.}\ \bibnamefont
  {C{\'{e}}leri}}, \bibinfo {author} {\bibfnamefont {Jayne}\ \bibnamefont
  {Thompson}}, \ and\ \bibinfo {author} {\bibfnamefont {Mile}\ \bibnamefont
  {Gu}},\ }\bibfield  {title} {\enquote {\bibinfo {title} {Fragile states are
  better for quantum metrology},}\ }\href {https://arxiv.org/abs/1608.01443}
  {\bibfield  {journal} {\bibinfo  {journal} {arXiv preprint arXiv:1608.01443}\
  } (\bibinfo {year} {2016})}\BibitemShut {NoStop}%
\bibitem [{\citenamefont {Zhang}\ \emph {et~al.}(2017)\citenamefont {Zhang},
  \citenamefont {Yadin}, \citenamefont {Hou}, \citenamefont {Cao},
  \citenamefont {Liu}, \citenamefont {Huang}, \citenamefont {Maity},
  \citenamefont {Vedral}, \citenamefont {Li}, \citenamefont {Guo},\ and\
  \citenamefont {Girolami}}]{girolami2017detecting}%
  \BibitemOpen
  \bibfield  {author} {\bibinfo {author} {\bibfnamefont {Chao}\ \bibnamefont
  {Zhang}}, \bibinfo {author} {\bibfnamefont {Benjamin}\ \bibnamefont {Yadin}},
  \bibinfo {author} {\bibfnamefont {Zhi-Bo}\ \bibnamefont {Hou}}, \bibinfo
  {author} {\bibfnamefont {Huan}\ \bibnamefont {Cao}}, \bibinfo {author}
  {\bibfnamefont {Bi-Heng}\ \bibnamefont {Liu}}, \bibinfo {author}
  {\bibfnamefont {Yun-Feng}\ \bibnamefont {Huang}}, \bibinfo {author}
  {\bibfnamefont {Reevu}\ \bibnamefont {Maity}}, \bibinfo {author}
  {\bibfnamefont {Vlatko}\ \bibnamefont {Vedral}}, \bibinfo {author}
  {\bibfnamefont {Chuan-Feng}\ \bibnamefont {Li}}, \bibinfo {author}
  {\bibfnamefont {Guang-Can}\ \bibnamefont {Guo}}, \ and\ \bibinfo {author}
  {\bibfnamefont {Davide}\ \bibnamefont {Girolami}},\ }\bibfield  {title}
  {\enquote {\bibinfo {title} {Detecting metrologically useful asymmetry and
  entanglement by a few local measurements},}\ }\href {\doibase
  10.1103/PhysRevA.96.042327} {\bibfield  {journal} {\bibinfo  {journal} {Phys.
  Rev. A}\ }\textbf {\bibinfo {volume} {96}},\ \bibinfo {pages} {042327}
  (\bibinfo {year} {2017})}\BibitemShut {NoStop}%
\bibitem [{\citenamefont {Toth}(2017)}]{toth2017lower}%
  \BibitemOpen
  \bibfield  {author} {\bibinfo {author} {\bibfnamefont {Geza}\ \bibnamefont
  {Toth}},\ }\bibfield  {title} {\enquote {\bibinfo {title} {Lower bounds on
  the quantum fisher information based on the variance and various types of
  entropies},}\ }\href {https://arxiv.org/abs/1701.07461} {\bibfield  {journal}
  {\bibinfo  {journal} {arXiv preprint arXiv:1701.07461}\ } (\bibinfo {year}
  {2017})}\BibitemShut {NoStop}%
\bibitem [{\citenamefont {Apellaniz}\ \emph {et~al.}(2017)\citenamefont
  {Apellaniz}, \citenamefont {Kleinmann}, \citenamefont {G{\"u}hne},\ and\
  \citenamefont {T{\'o}th}}]{apellaniz2017optimal}%
  \BibitemOpen
  \bibfield  {author} {\bibinfo {author} {\bibfnamefont {Iagoba}\ \bibnamefont
  {Apellaniz}}, \bibinfo {author} {\bibfnamefont {Matthias}\ \bibnamefont
  {Kleinmann}}, \bibinfo {author} {\bibfnamefont {Otfried}\ \bibnamefont
  {G{\"u}hne}}, \ and\ \bibinfo {author} {\bibfnamefont {G{\'e}za}\
  \bibnamefont {T{\'o}th}},\ }\bibfield  {title} {\enquote {\bibinfo {title}
  {Optimal witnessing of the quantum fisher information with few
  measurements},}\ }\href {\doibase 10.1103/PhysRevA.95.032330} {\bibfield
  {journal} {\bibinfo  {journal} {Physical Review A}\ }\textbf {\bibinfo
  {volume} {95}},\ \bibinfo {pages} {032330} (\bibinfo {year}
  {2017})}\BibitemShut {NoStop}%
\bibitem [{\citenamefont {Sone}\ \emph {et~al.}(2020)\citenamefont {Sone},
  \citenamefont {Cerezo}, \citenamefont {Beckey},\ and\ \citenamefont
  {Coles}}]{sone2020generalized}%
  \BibitemOpen
  \bibfield  {author} {\bibinfo {author} {\bibfnamefont {Akira}\ \bibnamefont
  {Sone}}, \bibinfo {author} {\bibfnamefont {M.}~\bibnamefont {Cerezo}},
  \bibinfo {author} {\bibfnamefont {Jacob~L}\ \bibnamefont {Beckey}}, \ and\
  \bibinfo {author} {\bibfnamefont {Patrick~J}\ \bibnamefont {Coles}},\
  }\bibfield  {title} {\enquote {\bibinfo {title} {A generalized measure of
  quantum fisher information},}\ }\href {https://arxiv.org/abs/2010.02904}
  {\bibfield  {journal} {\bibinfo  {journal} {arXiv preprint arXiv:2010.02904}\
  } (\bibinfo {year} {2020})}\BibitemShut {NoStop}%
\bibitem [{\citenamefont {Preskill}(2018)}]{preskill2018quantum}%
  \BibitemOpen
  \bibfield  {author} {\bibinfo {author} {\bibfnamefont {John}\ \bibnamefont
  {Preskill}},\ }\bibfield  {title} {\enquote {\bibinfo {title} {Quantum
  computing in the nisq era and beyond},}\ }\href {\doibase
  https://doi.org/10.22331/q-2018-08-06-79} {\bibfield  {journal} {\bibinfo
  {journal} {Quantum}\ }\textbf {\bibinfo {volume} {2}},\ \bibinfo {pages} {79}
  (\bibinfo {year} {2018})}\BibitemShut {NoStop}%
\bibitem [{\citenamefont {Cerezo}\ \emph
  {et~al.}(2020{\natexlab{b}})\citenamefont {Cerezo}, \citenamefont
  {Arrasmith}, \citenamefont {Babbush}, \citenamefont {Benjamin}, \citenamefont
  {Endo}, \citenamefont {Fujii}, \citenamefont {McClean}, \citenamefont
  {Mitarai}, \citenamefont {Yuan}, \citenamefont {Cincio},\ and\ \citenamefont
  {Coles}}]{cerezo2020variationalreview}%
  \BibitemOpen
  \bibfield  {author} {\bibinfo {author} {\bibfnamefont {M.}~\bibnamefont
  {Cerezo}}, \bibinfo {author} {\bibfnamefont {Andrew}\ \bibnamefont
  {Arrasmith}}, \bibinfo {author} {\bibfnamefont {Ryan}\ \bibnamefont
  {Babbush}}, \bibinfo {author} {\bibfnamefont {Simon~C}\ \bibnamefont
  {Benjamin}}, \bibinfo {author} {\bibfnamefont {Suguru}\ \bibnamefont {Endo}},
  \bibinfo {author} {\bibfnamefont {Keisuke}\ \bibnamefont {Fujii}}, \bibinfo
  {author} {\bibfnamefont {Jarrod~R}\ \bibnamefont {McClean}}, \bibinfo
  {author} {\bibfnamefont {Kosuke}\ \bibnamefont {Mitarai}}, \bibinfo {author}
  {\bibfnamefont {Xiao}\ \bibnamefont {Yuan}}, \bibinfo {author} {\bibfnamefont
  {Lukasz}\ \bibnamefont {Cincio}}, \ and\ \bibinfo {author} {\bibfnamefont
  {Patrick~J.}\ \bibnamefont {Coles}},\ }\bibfield  {title} {\enquote {\bibinfo
  {title} {Variational quantum algorithms},}\ }\href
  {https://arxiv.org/abs/2012.09265} {\bibfield  {journal} {\bibinfo  {journal}
  {arXiv preprint arXiv:2012.09265}\ } (\bibinfo {year}
  {2020}{\natexlab{b}})}\BibitemShut {NoStop}%
\bibitem [{\citenamefont {Endo}\ \emph {et~al.}(2021)\citenamefont {Endo},
  \citenamefont {Cai}, \citenamefont {Benjamin},\ and\ \citenamefont
  {Yuan}}]{endo2021hybrid}%
  \BibitemOpen
  \bibfield  {author} {\bibinfo {author} {\bibfnamefont {Suguru}\ \bibnamefont
  {Endo}}, \bibinfo {author} {\bibfnamefont {Zhenyu}\ \bibnamefont {Cai}},
  \bibinfo {author} {\bibfnamefont {Simon~C}\ \bibnamefont {Benjamin}}, \ and\
  \bibinfo {author} {\bibfnamefont {Xiao}\ \bibnamefont {Yuan}},\ }\bibfield
  {title} {\enquote {\bibinfo {title} {Hybrid quantum-classical algorithms and
  quantum error mitigation},}\ }\href {\doibase 10.7566/JPSJ.90.032001}
  {\bibfield  {journal} {\bibinfo  {journal} {Journal of the Physical Society
  of Japan}\ }\textbf {\bibinfo {volume} {90}},\ \bibinfo {pages} {032001}
  (\bibinfo {year} {2021})}\BibitemShut {NoStop}%
\bibitem [{\citenamefont {Bharti}\ \emph {et~al.}(2021)\citenamefont {Bharti},
  \citenamefont {Cervera-Lierta}, \citenamefont {Kyaw}, \citenamefont {Haug},
  \citenamefont {Alperin-Lea}, \citenamefont {Anand}, \citenamefont {Degroote},
  \citenamefont {Heimonen}, \citenamefont {Kottmann}, \citenamefont {Menke},
  \citenamefont {Mok}, \citenamefont {Sim}, \citenamefont {Kwek},\ and\
  \citenamefont {Aspuru-Guzik}}]{bharti2021noisy}%
  \BibitemOpen
  \bibfield  {author} {\bibinfo {author} {\bibfnamefont {Kishor}\ \bibnamefont
  {Bharti}}, \bibinfo {author} {\bibfnamefont {Alba}\ \bibnamefont
  {Cervera-Lierta}}, \bibinfo {author} {\bibfnamefont {Thi~Ha}\ \bibnamefont
  {Kyaw}}, \bibinfo {author} {\bibfnamefont {Tobias}\ \bibnamefont {Haug}},
  \bibinfo {author} {\bibfnamefont {Sumner}\ \bibnamefont {Alperin-Lea}},
  \bibinfo {author} {\bibfnamefont {Abhinav}\ \bibnamefont {Anand}}, \bibinfo
  {author} {\bibfnamefont {Matthias}\ \bibnamefont {Degroote}}, \bibinfo
  {author} {\bibfnamefont {Hermanni}\ \bibnamefont {Heimonen}}, \bibinfo
  {author} {\bibfnamefont {Jakob~S.}\ \bibnamefont {Kottmann}}, \bibinfo
  {author} {\bibfnamefont {Tim}\ \bibnamefont {Menke}}, \bibinfo {author}
  {\bibfnamefont {Wai-Keong}\ \bibnamefont {Mok}}, \bibinfo {author}
  {\bibfnamefont {Sukin}\ \bibnamefont {Sim}}, \bibinfo {author} {\bibfnamefont
  {Leong-Chuan}\ \bibnamefont {Kwek}}, \ and\ \bibinfo {author} {\bibfnamefont
  {Alán}\ \bibnamefont {Aspuru-Guzik}},\ }\bibfield  {title} {\enquote
  {\bibinfo {title} {Noisy intermediate-scale quantum (nisq) algorithms},}\
  }\href {https://arxiv.org/abs/2101.08448} {\bibfield  {journal} {\bibinfo
  {journal} {arXiv preprint arXiv:2101.08448}\ } (\bibinfo {year}
  {2021})}\BibitemShut {NoStop}%
\bibitem [{\citenamefont {Cerezo}\ \emph
  {et~al.}(2020{\natexlab{c}})\citenamefont {Cerezo}, \citenamefont {Sharma},
  \citenamefont {Arrasmith},\ and\ \citenamefont
  {Coles}}]{cerezo2020variational}%
  \BibitemOpen
  \bibfield  {author} {\bibinfo {author} {\bibfnamefont {M.}~\bibnamefont
  {Cerezo}}, \bibinfo {author} {\bibfnamefont {Kunal}\ \bibnamefont {Sharma}},
  \bibinfo {author} {\bibfnamefont {Andrew}\ \bibnamefont {Arrasmith}}, \ and\
  \bibinfo {author} {\bibfnamefont {Patrick~J}\ \bibnamefont {Coles}},\
  }\bibfield  {title} {\enquote {\bibinfo {title} {Variational quantum state
  eigensolver},}\ }\href {https://arxiv.org/abs/2004.01372} {\bibfield
  {journal} {\bibinfo  {journal} {arXiv preprint arXiv:2004.01372}\ } (\bibinfo
  {year} {2020}{\natexlab{c}})}\BibitemShut {NoStop}%
\bibitem [{\citenamefont {Bravo-Prieto}\ \emph {et~al.}(2019)\citenamefont
  {Bravo-Prieto}, \citenamefont {LaRose}, \citenamefont {Cerezo}, \citenamefont
  {Subasi}, \citenamefont {Cincio},\ and\ \citenamefont
  {Coles}}]{bravo2020variational}%
  \BibitemOpen
  \bibfield  {author} {\bibinfo {author} {\bibfnamefont {Carlos}\ \bibnamefont
  {Bravo-Prieto}}, \bibinfo {author} {\bibfnamefont {Ryan}\ \bibnamefont
  {LaRose}}, \bibinfo {author} {\bibfnamefont {M.}~\bibnamefont {Cerezo}},
  \bibinfo {author} {\bibfnamefont {Yigit}\ \bibnamefont {Subasi}}, \bibinfo
  {author} {\bibfnamefont {Lukasz}\ \bibnamefont {Cincio}}, \ and\ \bibinfo
  {author} {\bibfnamefont {Patrick}\ \bibnamefont {Coles}},\ }\bibfield
  {title} {\enquote {\bibinfo {title} {Variational quantum linear solver},}\
  }\href {https://arxiv.org/abs/1909.05820} {\bibfield  {journal} {\bibinfo
  {journal} {arXiv preprint arXiv:1909.05820}\ } (\bibinfo {year}
  {2019})}\BibitemShut {NoStop}%
\bibitem [{\citenamefont {Koczor}\ \emph {et~al.}(2020)\citenamefont {Koczor},
  \citenamefont {Endo}, \citenamefont {Jones}, \citenamefont {Matsuzaki},\ and\
  \citenamefont {Benjamin}}]{koczor2020variational}%
  \BibitemOpen
  \bibfield  {author} {\bibinfo {author} {\bibfnamefont {B{\'a}lint}\
  \bibnamefont {Koczor}}, \bibinfo {author} {\bibfnamefont {Suguru}\
  \bibnamefont {Endo}}, \bibinfo {author} {\bibfnamefont {Tyson}\ \bibnamefont
  {Jones}}, \bibinfo {author} {\bibfnamefont {Yuichiro}\ \bibnamefont
  {Matsuzaki}}, \ and\ \bibinfo {author} {\bibfnamefont {Simon~C}\ \bibnamefont
  {Benjamin}},\ }\bibfield  {title} {\enquote {\bibinfo {title}
  {Variational-state quantum metrology},}\ }\href {\doibase
  10.1088/1367-2630/ab965e} {\bibfield  {journal} {\bibinfo  {journal} {New
  Journal of Physics}\ } (\bibinfo {year} {2020}),\
  10.1088/1367-2630/ab965e}\BibitemShut {NoStop}%
\bibitem [{\citenamefont {Yang}\ \emph {et~al.}(2020)\citenamefont {Yang},
  \citenamefont {Thompson}, \citenamefont {Wu}, \citenamefont {Gu},
  \citenamefont {Peng},\ and\ \citenamefont {Du}}]{yang2020probe}%
  \BibitemOpen
  \bibfield  {author} {\bibinfo {author} {\bibfnamefont {Xiaodong}\
  \bibnamefont {Yang}}, \bibinfo {author} {\bibfnamefont {Jayne}\ \bibnamefont
  {Thompson}}, \bibinfo {author} {\bibfnamefont {Ze}~\bibnamefont {Wu}},
  \bibinfo {author} {\bibfnamefont {Mile}\ \bibnamefont {Gu}}, \bibinfo
  {author} {\bibfnamefont {Xinhua}\ \bibnamefont {Peng}}, \ and\ \bibinfo
  {author} {\bibfnamefont {Jiangfeng}\ \bibnamefont {Du}},\ }\bibfield  {title}
  {\enquote {\bibinfo {title} {Probe optimization for quantum metrology via
  closed-loop learning control},}\ }\href {\doibase 10.1038/s41534-020-00292-z}
  {\bibfield  {journal} {\bibinfo  {journal} {npj Quantum Inf}\ }\textbf
  {\bibinfo {volume} {6}},\ \bibinfo {pages} {62} (\bibinfo {year}
  {2020})}\BibitemShut {NoStop}%
\bibitem [{\citenamefont {Meyer}\ \emph {et~al.}(2020)\citenamefont {Meyer},
  \citenamefont {Borregaard},\ and\ \citenamefont
  {Eisert}}]{meyer2020variational}%
  \BibitemOpen
  \bibfield  {author} {\bibinfo {author} {\bibfnamefont {Johannes~Jakob}\
  \bibnamefont {Meyer}}, \bibinfo {author} {\bibfnamefont {Johannes}\
  \bibnamefont {Borregaard}}, \ and\ \bibinfo {author} {\bibfnamefont {Jens}\
  \bibnamefont {Eisert}},\ }\bibfield  {title} {\enquote {\bibinfo {title} {A
  variational toolbox for quantum multi-parameter estimation},}\ }\href
  {https://arxiv.org/abs/2006.06303} {\bibfield  {journal} {\bibinfo  {journal}
  {arXiv preprint arXiv:2006.06303}\ } (\bibinfo {year} {2020})}\BibitemShut
  {NoStop}%
\bibitem [{\citenamefont {Beckey}\ \emph {et~al.}(2020)\citenamefont {Beckey},
  \citenamefont {Cerezo}, \citenamefont {Sone},\ and\ \citenamefont
  {Coles}}]{beckey2020variational}%
  \BibitemOpen
  \bibfield  {author} {\bibinfo {author} {\bibfnamefont {Jacob~L}\ \bibnamefont
  {Beckey}}, \bibinfo {author} {\bibfnamefont {M.}~\bibnamefont {Cerezo}},
  \bibinfo {author} {\bibfnamefont {Akira}\ \bibnamefont {Sone}}, \ and\
  \bibinfo {author} {\bibfnamefont {Patrick~J}\ \bibnamefont {Coles}},\
  }\bibfield  {title} {\enquote {\bibinfo {title} {Variational quantum
  algorithm for estimating the quantum fisher information},}\ }\href
  {https://arxiv.org/abs/2010.10488} {\bibfield  {journal} {\bibinfo  {journal}
  {arXiv preprint arXiv:2010.10488}\ } (\bibinfo {year} {2020})}\BibitemShut
  {NoStop}%
\bibitem [{\citenamefont {Ma}\ \emph {et~al.}(2020)\citenamefont {Ma},
  \citenamefont {Gokhale}, \citenamefont {Zheng}, \citenamefont {Zhou},
  \citenamefont {Yu}, \citenamefont {Jiang}, \citenamefont {Maurer},\ and\
  \citenamefont {Chong}}]{ma2020adaptive}%
  \BibitemOpen
  \bibfield  {author} {\bibinfo {author} {\bibfnamefont {Ziqi}\ \bibnamefont
  {Ma}}, \bibinfo {author} {\bibfnamefont {Pranav}\ \bibnamefont {Gokhale}},
  \bibinfo {author} {\bibfnamefont {Tian-Xing}\ \bibnamefont {Zheng}}, \bibinfo
  {author} {\bibfnamefont {Sisi}\ \bibnamefont {Zhou}}, \bibinfo {author}
  {\bibfnamefont {Xiaofei}\ \bibnamefont {Yu}}, \bibinfo {author}
  {\bibfnamefont {Liang}\ \bibnamefont {Jiang}}, \bibinfo {author}
  {\bibfnamefont {Peter}\ \bibnamefont {Maurer}}, \ and\ \bibinfo {author}
  {\bibfnamefont {Frederic~T}\ \bibnamefont {Chong}},\ }\bibfield  {title}
  {\enquote {\bibinfo {title} {Adaptive circuit learning for quantum
  metrology},}\ }\href {https://arxiv.org/abs/2010.08702} {\bibfield  {journal}
  {\bibinfo  {journal} {arXiv preprint arXiv:2010.08702}\ } (\bibinfo {year}
  {2020})}\BibitemShut {NoStop}%
\bibitem [{\citenamefont {Braunstein}\ and\ \citenamefont
  {Caves}(1994)}]{braunstein1994statistical}%
  \BibitemOpen
  \bibfield  {author} {\bibinfo {author} {\bibfnamefont {Samuel~L}\
  \bibnamefont {Braunstein}}\ and\ \bibinfo {author} {\bibfnamefont
  {Carlton~M}\ \bibnamefont {Caves}},\ }\bibfield  {title} {\enquote {\bibinfo
  {title} {Statistical distance and the geometry of quantum states},}\ }\href
  {\doibase 10.1103/PhysRevLett.72.3439} {\bibfield  {journal} {\bibinfo
  {journal} {Physical Review Letters}\ }\textbf {\bibinfo {volume} {72}},\
  \bibinfo {pages} {3439} (\bibinfo {year} {1994})}\BibitemShut {NoStop}%
\bibitem [{\citenamefont {Liu}\ \emph {et~al.}(2014)\citenamefont {Liu},
  \citenamefont {Xiong}, \citenamefont {Song},\ and\ \citenamefont
  {Wang}}]{liu2014fidelity}%
  \BibitemOpen
  \bibfield  {author} {\bibinfo {author} {\bibfnamefont {Jing}\ \bibnamefont
  {Liu}}, \bibinfo {author} {\bibfnamefont {Heng-Na}\ \bibnamefont {Xiong}},
  \bibinfo {author} {\bibfnamefont {Fei}\ \bibnamefont {Song}}, \ and\ \bibinfo
  {author} {\bibfnamefont {Xiaoguang}\ \bibnamefont {Wang}},\ }\bibfield
  {title} {\enquote {\bibinfo {title} {Fidelity susceptibility and quantum
  fisher information for density operators with arbitrary ranks},}\ }\href
  {\doibase 10.1016/j.physa.2014.05.028} {\bibfield  {journal} {\bibinfo
  {journal} {Physica A: Statistical Mechanics and its Applications}\ }\textbf
  {\bibinfo {volume} {410}},\ \bibinfo {pages} {167--173} (\bibinfo {year}
  {2014})}\BibitemShut {NoStop}%
\bibitem [{\citenamefont {Hayashi}(2004{\natexlab{b}})}]{Hayashi2004Quantum}%
  \BibitemOpen
  \bibfield  {author} {\bibinfo {author} {\bibfnamefont {Masahito}\
  \bibnamefont {Hayashi}},\ }\href {\doibase 10.1007/978-3-662-49725-8} {\emph
  {\bibinfo {title} {Quantum Information Theory: Mathematical Foundation (2nd
  edition)}}}\ (\bibinfo  {publisher} {Springer},\ \bibinfo {year}
  {2004})\BibitemShut {NoStop}%
\bibitem [{\citenamefont {Miszczak}\ \emph {et~al.}(2009)\citenamefont
  {Miszczak}, \citenamefont {Puchala}, \citenamefont {Horodecki}, \citenamefont
  {Uhlmann},\ and\ \citenamefont {Zyczkowski}}]{Miszczak2009sub}%
  \BibitemOpen
  \bibfield  {author} {\bibinfo {author} {\bibfnamefont {Jaroslaw~Adam}\
  \bibnamefont {Miszczak}}, \bibinfo {author} {\bibfnamefont {Zbigniew}\
  \bibnamefont {Puchala}}, \bibinfo {author} {\bibfnamefont {Pawel}\
  \bibnamefont {Horodecki}}, \bibinfo {author} {\bibfnamefont {Armin}\
  \bibnamefont {Uhlmann}}, \ and\ \bibinfo {author} {\bibfnamefont {Karol}\
  \bibnamefont {Zyczkowski}},\ }\bibfield  {title} {\enquote {\bibinfo {title}
  {Sub- and super-fidelity as bounds for quantum fidelity},}\ }\href {\doibase
  https://doi.org/10.26421/QIC9.1-2} {\bibfield  {journal} {\bibinfo  {journal}
  {Quantum Information {\&} Computation}\ }\textbf {\bibinfo {volume} {9}},\
  \bibinfo {pages} {103--130} (\bibinfo {year} {2009})}\BibitemShut {NoStop}%
\bibitem [{\citenamefont {Mendon{\c{c}}a}\ \emph {et~al.}(2008)\citenamefont
  {Mendon{\c{c}}a}, \citenamefont {Napolitano}, \citenamefont {Marchiolli},
  \citenamefont {Foster},\ and\ \citenamefont
  {Liang}}]{mendoncca2008alternative}%
  \BibitemOpen
  \bibfield  {author} {\bibinfo {author} {\bibfnamefont {Paulo~EMF}\
  \bibnamefont {Mendon{\c{c}}a}}, \bibinfo {author} {\bibfnamefont {Reginaldo
  d~J}\ \bibnamefont {Napolitano}}, \bibinfo {author} {\bibfnamefont
  {Marcelo~A}\ \bibnamefont {Marchiolli}}, \bibinfo {author} {\bibfnamefont
  {Christopher~J}\ \bibnamefont {Foster}}, \ and\ \bibinfo {author}
  {\bibfnamefont {Yeong-Cherng}\ \bibnamefont {Liang}},\ }\bibfield  {title}
  {\enquote {\bibinfo {title} {Alternative fidelity measure between quantum
  states},}\ }\href {\doibase 10.1103/PhysRevA.78.052330} {\bibfield  {journal}
  {\bibinfo  {journal} {Physical Review A}\ }\textbf {\bibinfo {volume} {78}},\
  \bibinfo {pages} {052330} (\bibinfo {year} {2008})}\BibitemShut {NoStop}%
\bibitem [{\citenamefont {Pucha{\l}a}\ and\ \citenamefont
  {Miszczak}(2009)}]{puchala2009bound}%
  \BibitemOpen
  \bibfield  {author} {\bibinfo {author} {\bibfnamefont {Zbigniew}\
  \bibnamefont {Pucha{\l}a}}\ and\ \bibinfo {author} {\bibfnamefont
  {Jaros{\l}aw~Adam}\ \bibnamefont {Miszczak}},\ }\bibfield  {title} {\enquote
  {\bibinfo {title} {Bound on trace distance based on superfidelity},}\ }\href
  {\doibase 10.1103/PhysRevA.79.024302} {\bibfield  {journal} {\bibinfo
  {journal} {Physical Review A}\ }\textbf {\bibinfo {volume} {79}},\ \bibinfo
  {pages} {024302} (\bibinfo {year} {2009})}\BibitemShut {NoStop}%
\bibitem [{\citenamefont {Cincio}\ \emph {et~al.}(2018)\citenamefont {Cincio},
  \citenamefont {Suba{\c{s}}{\i}}, \citenamefont {Sornborger},\ and\
  \citenamefont {Coles}}]{cincio2018learning}%
  \BibitemOpen
  \bibfield  {author} {\bibinfo {author} {\bibfnamefont {Lukasz}\ \bibnamefont
  {Cincio}}, \bibinfo {author} {\bibfnamefont {Yi{\u{g}}it}\ \bibnamefont
  {Suba{\c{s}}{\i}}}, \bibinfo {author} {\bibfnamefont {Andrew~T}\ \bibnamefont
  {Sornborger}}, \ and\ \bibinfo {author} {\bibfnamefont {Patrick~J}\
  \bibnamefont {Coles}},\ }\bibfield  {title} {\enquote {\bibinfo {title}
  {Learning the quantum algorithm for state overlap},}\ }\href {\doibase
  10.1088/1367-2630/aae94a} {\bibfield  {journal} {\bibinfo  {journal} {New
  Journal of Physics}\ }\textbf {\bibinfo {volume} {20}},\ \bibinfo {pages}
  {113022} (\bibinfo {year} {2018})}\BibitemShut {NoStop}%
\bibitem [{\citenamefont {G{\"a}rttner}\ \emph {et~al.}(2018)\citenamefont
  {G{\"a}rttner}, \citenamefont {Hauke},\ and\ \citenamefont
  {Rey}}]{garttner2018relating}%
  \BibitemOpen
  \bibfield  {author} {\bibinfo {author} {\bibfnamefont {Martin}\ \bibnamefont
  {G{\"a}rttner}}, \bibinfo {author} {\bibfnamefont {Philipp}\ \bibnamefont
  {Hauke}}, \ and\ \bibinfo {author} {\bibfnamefont {Ana~Maria}\ \bibnamefont
  {Rey}},\ }\bibfield  {title} {\enquote {\bibinfo {title} {Relating
  out-of-time-order correlations to entanglement via multiple-quantum
  coherences},}\ }\href {\doibase 10.1103/PhysRevLett.120.040402} {\bibfield
  {journal} {\bibinfo  {journal} {Physical Review Letters}\ }\textbf {\bibinfo
  {volume} {120}},\ \bibinfo {pages} {040402} (\bibinfo {year}
  {2018})}\BibitemShut {NoStop}%
\bibitem [{\citenamefont {Rivas}\ and\ \citenamefont
  {Luis}(2008)}]{rivas2008intrinsic}%
  \BibitemOpen
  \bibfield  {author} {\bibinfo {author} {\bibfnamefont {\'Angel}\ \bibnamefont
  {Rivas}}\ and\ \bibinfo {author} {\bibfnamefont {Alfredo}\ \bibnamefont
  {Luis}},\ }\bibfield  {title} {\enquote {\bibinfo {title} {Intrinsic
  metrological resolution as a distance measure and nonclassical light},}\
  }\href {\doibase 10.1103/PhysRevA.77.063813} {\bibfield  {journal} {\bibinfo
  {journal} {Phys. Rev. A}\ }\textbf {\bibinfo {volume} {77}},\ \bibinfo
  {pages} {063813} (\bibinfo {year} {2008})}\BibitemShut {NoStop}%
\bibitem [{\citenamefont {Rivas}\ and\ \citenamefont
  {Luis}(2010)}]{rivas2010precision}%
  \BibitemOpen
  \bibfield  {author} {\bibinfo {author} {\bibfnamefont {\'Angel}\ \bibnamefont
  {Rivas}}\ and\ \bibinfo {author} {\bibfnamefont {Alfredo}\ \bibnamefont
  {Luis}},\ }\bibfield  {title} {\enquote {\bibinfo {title} {Precision quantum
  metrology and nonclassicality in linear and nonlinear detection schemes},}\
  }\href {\doibase 10.1103/PhysRevLett.105.010403} {\bibfield  {journal}
  {\bibinfo  {journal} {Phys. Rev. Lett.}\ }\textbf {\bibinfo {volume} {105}},\
  \bibinfo {pages} {010403} (\bibinfo {year} {2010})}\BibitemShut {NoStop}%
\bibitem [{\citenamefont {Luo}\ and\ \citenamefont {Sun}(2020)}]{luo2020skew}%
  \BibitemOpen
  \bibfield  {author} {\bibinfo {author} {\bibfnamefont {Shunlong}\
  \bibnamefont {Luo}}\ and\ \bibinfo {author} {\bibfnamefont {Yuan}\
  \bibnamefont {Sun}},\ }\bibfield  {title} {\enquote {\bibinfo {title} {Skew
  information revisited: Its variants and a comparison of them},}\ }\href
  {\doibase 10.1134/S0040577920010092} {\bibfield  {journal} {\bibinfo
  {journal} {Theoretical and Mathematical Physics}\ }\textbf {\bibinfo {volume}
  {202}},\ \bibinfo {pages} {104--111} (\bibinfo {year} {2020})}\BibitemShut
  {NoStop}%
\bibitem [{\citenamefont {Alipour}\ and\ \citenamefont
  {Rezakhani}(2015)}]{alipour2015extended}%
  \BibitemOpen
  \bibfield  {author} {\bibinfo {author} {\bibfnamefont {S.}~\bibnamefont
  {Alipour}}\ and\ \bibinfo {author} {\bibfnamefont {A.~T.}\ \bibnamefont
  {Rezakhani}},\ }\bibfield  {title} {\enquote {\bibinfo {title} {Extended
  convexity of quantum fisher information in quantum metrology},}\ }\href
  {\doibase 10.1103/PhysRevA.91.042104} {\bibfield  {journal} {\bibinfo
  {journal} {Phys. Rev. A}\ }\textbf {\bibinfo {volume} {91}},\ \bibinfo
  {pages} {042104} (\bibinfo {year} {2015})}\BibitemShut {NoStop}%
\bibitem [{\citenamefont {Paris}(2009)}]{paris2009quantum}%
  \BibitemOpen
  \bibfield  {author} {\bibinfo {author} {\bibfnamefont {Matteo~GA}\
  \bibnamefont {Paris}},\ }\bibfield  {title} {\enquote {\bibinfo {title}
  {Quantum estimation for quantum technology},}\ }\href {\doibase
  10.1142/S0219749909004839} {\bibfield  {journal} {\bibinfo  {journal}
  {International Journal of Quantum Information}\ }\textbf {\bibinfo {volume}
  {7}},\ \bibinfo {pages} {125--137} (\bibinfo {year} {2009})}\BibitemShut
  {NoStop}%
\bibitem [{\citenamefont {Liu}\ \emph {et~al.}(2016)\citenamefont {Liu},
  \citenamefont {Chen}, \citenamefont {Jing},\ and\ \citenamefont
  {Wang}}]{liu2016quantum}%
  \BibitemOpen
  \bibfield  {author} {\bibinfo {author} {\bibfnamefont {Jing}\ \bibnamefont
  {Liu}}, \bibinfo {author} {\bibfnamefont {Jie}\ \bibnamefont {Chen}},
  \bibinfo {author} {\bibfnamefont {Xiao-Xing}\ \bibnamefont {Jing}}, \ and\
  \bibinfo {author} {\bibfnamefont {Xiaoguang}\ \bibnamefont {Wang}},\
  }\bibfield  {title} {\enquote {\bibinfo {title} {Quantum fisher information
  and symmetric logarithmic derivative via anti-commutators},}\ }\href
  {\doibase 10.1088/1751-8113/49/27/275302} {\bibfield  {journal} {\bibinfo
  {journal} {Journal of Physics A: Mathematical and Theoretical}\ }\textbf
  {\bibinfo {volume} {49}},\ \bibinfo {pages} {275302} (\bibinfo {year}
  {2016})}\BibitemShut {NoStop}%
\bibitem [{\citenamefont {Fiderer}\ \emph {et~al.}(2019)\citenamefont
  {Fiderer}, \citenamefont {Fra{\"\i}sse},\ and\ \citenamefont
  {Braun}}]{fiderer2019maximal}%
  \BibitemOpen
  \bibfield  {author} {\bibinfo {author} {\bibfnamefont {Lukas~J}\ \bibnamefont
  {Fiderer}}, \bibinfo {author} {\bibfnamefont {Julien~ME}\ \bibnamefont
  {Fra{\"\i}sse}}, \ and\ \bibinfo {author} {\bibfnamefont {Daniel}\
  \bibnamefont {Braun}},\ }\bibfield  {title} {\enquote {\bibinfo {title}
  {Maximal quantum fisher information for mixed states},}\ }\href {\doibase
  10.1103/PhysRevLett.123.250502} {\bibfield  {journal} {\bibinfo  {journal}
  {Physical Review Letters}\ }\textbf {\bibinfo {volume} {123}},\ \bibinfo
  {pages} {250502} (\bibinfo {year} {2019})}\BibitemShut {NoStop}%
\bibitem [{\citenamefont {H{\"u}bner}(1992)}]{hubner1992explicit}%
  \BibitemOpen
  \bibfield  {author} {\bibinfo {author} {\bibfnamefont {Matthias}\
  \bibnamefont {H{\"u}bner}},\ }\bibfield  {title} {\enquote {\bibinfo {title}
  {Explicit computation of the bures distance for density matrices},}\ }\href
  {\doibase 10.1134/S0040577920010092} {\bibfield  {journal} {\bibinfo
  {journal} {Physics Letters. A}\ }\textbf {\bibinfo {volume} {163}},\ \bibinfo
  {pages} {239--242} (\bibinfo {year} {1992})}\BibitemShut {NoStop}%
\bibitem [{\citenamefont
  {{\v{S}}afr{\'a}nek}(2017)}]{vsafranek2017discontinuities}%
  \BibitemOpen
  \bibfield  {author} {\bibinfo {author} {\bibfnamefont {Dominik}\ \bibnamefont
  {{\v{S}}afr{\'a}nek}},\ }\bibfield  {title} {\enquote {\bibinfo {title}
  {Discontinuities of the quantum fisher information and the bures metric},}\
  }\href {\doibase 10.1103/PhysRevA.95.052320} {\bibfield  {journal} {\bibinfo
  {journal} {Physical Review A}\ }\textbf {\bibinfo {volume} {95}},\ \bibinfo
  {pages} {052320} (\bibinfo {year} {2017})}\BibitemShut {NoStop}%
\bibitem [{\citenamefont {Luo}\ and\ \citenamefont
  {Zhang}(2004)}]{luo2004informational}%
  \BibitemOpen
  \bibfield  {author} {\bibinfo {author} {\bibfnamefont {Shunlong}\
  \bibnamefont {Luo}}\ and\ \bibinfo {author} {\bibfnamefont {Qiang}\
  \bibnamefont {Zhang}},\ }\bibfield  {title} {\enquote {\bibinfo {title}
  {Informational distance on quantum-state space},}\ }\href {\doibase
  10.1103/PhysRevA.69.032106} {\bibfield  {journal} {\bibinfo  {journal}
  {Physical Review A}\ }\textbf {\bibinfo {volume} {69}},\ \bibinfo {pages}
  {032106} (\bibinfo {year} {2004})}\BibitemShut {NoStop}%
\bibitem [{\citenamefont {Wigner}\ and\ \citenamefont
  {Yanase}(1997)}]{wigner1997information}%
  \BibitemOpen
  \bibfield  {author} {\bibinfo {author} {\bibfnamefont {Eugene~P}\
  \bibnamefont {Wigner}}\ and\ \bibinfo {author} {\bibfnamefont {Mutsuo~M}\
  \bibnamefont {Yanase}},\ }\bibfield  {title} {\enquote {\bibinfo {title}
  {Information contents of distributions},}\ }in\ \href {\doibase
  10.1073/pnas.49.6.910} {\emph {\bibinfo {booktitle} {Part I: Particles and
  Fields. Part II: Foundations of Quantum Mechanics}}}\ (\bibinfo  {publisher}
  {Springer},\ \bibinfo {year} {1997})\ pp.\ \bibinfo {pages}
  {452--460}\BibitemShut {NoStop}%
\bibitem [{\citenamefont {Takagi}(2019)}]{takagi2019skew}%
  \BibitemOpen
  \bibfield  {author} {\bibinfo {author} {\bibfnamefont {Ryuji}\ \bibnamefont
  {Takagi}},\ }\bibfield  {title} {\enquote {\bibinfo {title} {Skew
  informations from an operational view via resource theory of asymmetry},}\
  }\href {\doibase 10.1038/s41598-019-50279-w} {\bibfield  {journal} {\bibinfo
  {journal} {Scientific Reports}\ }\textbf {\bibinfo {volume} {9}},\ \bibinfo
  {pages} {1--12} (\bibinfo {year} {2019})}\BibitemShut {NoStop}%
\bibitem [{\citenamefont {Luo}\ and\ \citenamefont
  {Sun}(2017)}]{luo2017quantum}%
  \BibitemOpen
  \bibfield  {author} {\bibinfo {author} {\bibfnamefont {Shunlong}\
  \bibnamefont {Luo}}\ and\ \bibinfo {author} {\bibfnamefont {Yuan}\
  \bibnamefont {Sun}},\ }\bibfield  {title} {\enquote {\bibinfo {title}
  {Quantum coherence versus quantum uncertainty},}\ }\href {\doibase
  10.1103/PhysRevA.96.022130} {\bibfield  {journal} {\bibinfo  {journal} {Phys.
  Rev. A}\ }\textbf {\bibinfo {volume} {96}},\ \bibinfo {pages} {022130}
  (\bibinfo {year} {2017})}\BibitemShut {NoStop}%
\bibitem [{\citenamefont {Luo}(2003)}]{luo2003wigner}%
  \BibitemOpen
  \bibfield  {author} {\bibinfo {author} {\bibfnamefont {Shunlong}\
  \bibnamefont {Luo}},\ }\bibfield  {title} {\enquote {\bibinfo {title}
  {Wigner-yanase skew information vs. quantum fisher information},}\ }\href
  {\doibase 10.1090/s0002-9939-03-07175-2} {\bibfield  {journal} {\bibinfo
  {journal} {Proceedings of the American Mathematical Society}\ }\textbf
  {\bibinfo {volume} {132}},\ \bibinfo {pages} {885–890} (\bibinfo {year}
  {2003})}\BibitemShut {NoStop}%
\bibitem [{\citenamefont {Du}\ and\ \citenamefont {Bai}(2015)}]{du2015wigner}%
  \BibitemOpen
  \bibfield  {author} {\bibinfo {author} {\bibfnamefont {Shuanping}\
  \bibnamefont {Du}}\ and\ \bibinfo {author} {\bibfnamefont {Zhaofang}\
  \bibnamefont {Bai}},\ }\bibfield  {title} {\enquote {\bibinfo {title} {The
  wigner--yanase information can increase under phase sensitive incoherent
  operations},}\ }\href {\doibase 10.1016/j.aop.2015.04.023} {\bibfield
  {journal} {\bibinfo  {journal} {Annals of Physics}\ }\textbf {\bibinfo
  {volume} {359}},\ \bibinfo {pages} {136--140} (\bibinfo {year}
  {2015})}\BibitemShut {NoStop}%
\bibitem [{\citenamefont {Marvian}\ and\ \citenamefont
  {Spekkens}(2016)}]{marvian2016quantify}%
  \BibitemOpen
  \bibfield  {author} {\bibinfo {author} {\bibfnamefont {Iman}\ \bibnamefont
  {Marvian}}\ and\ \bibinfo {author} {\bibfnamefont {Robert~W}\ \bibnamefont
  {Spekkens}},\ }\bibfield  {title} {\enquote {\bibinfo {title} {How to
  quantify coherence: Distinguishing speakable and unspeakable notions},}\
  }\href {\doibase 10.1103/PhysRevA.94.052324} {\bibfield  {journal} {\bibinfo
  {journal} {Physical Review A}\ }\textbf {\bibinfo {volume} {94}},\ \bibinfo
  {pages} {052324} (\bibinfo {year} {2016})}\BibitemShut {NoStop}%
\bibitem [{\citenamefont {Yadin}\ and\ \citenamefont
  {Vedral}(2016)}]{yadin2016general}%
  \BibitemOpen
  \bibfield  {author} {\bibinfo {author} {\bibfnamefont {Benjamin}\
  \bibnamefont {Yadin}}\ and\ \bibinfo {author} {\bibfnamefont {Vlatko}\
  \bibnamefont {Vedral}},\ }\bibfield  {title} {\enquote {\bibinfo {title}
  {General framework for quantum macroscopicity in terms of coherence},}\
  }\href {\doibase 10.1103/PhysRevA.93.022122} {\bibfield  {journal} {\bibinfo
  {journal} {Physical Review A}\ }\textbf {\bibinfo {volume} {93}},\ \bibinfo
  {pages} {022122} (\bibinfo {year} {2016})}\BibitemShut {NoStop}%
\bibitem [{\citenamefont {Klimov}\ \emph {et~al.}(2017)\citenamefont {Klimov},
  \citenamefont {Zwierz}, \citenamefont {Wallentowitz}, \citenamefont
  {Jarzyna},\ and\ \citenamefont {Banaszek}}]{klimov2017optimal}%
  \BibitemOpen
  \bibfield  {author} {\bibinfo {author} {\bibfnamefont {Andrei~B}\
  \bibnamefont {Klimov}}, \bibinfo {author} {\bibfnamefont {Marcin}\
  \bibnamefont {Zwierz}}, \bibinfo {author} {\bibfnamefont {Sascha}\
  \bibnamefont {Wallentowitz}}, \bibinfo {author} {\bibfnamefont {Marcin}\
  \bibnamefont {Jarzyna}}, \ and\ \bibinfo {author} {\bibfnamefont {Konrad}\
  \bibnamefont {Banaszek}},\ }\bibfield  {title} {\enquote {\bibinfo {title}
  {Optimal lossy quantum interferometry in phase space},}\ }\href {\doibase
  10.1088/1367-2630/aa73b0} {\bibfield  {journal} {\bibinfo  {journal} {New
  Journal of Physics}\ }\textbf {\bibinfo {volume} {19}},\ \bibinfo {pages}
  {073013} (\bibinfo {year} {2017})}\BibitemShut {NoStop}%
\bibitem [{\citenamefont {Hradil}\ \emph {et~al.}(2019)\citenamefont {Hradil},
  \citenamefont {Řeháček}, \citenamefont {Sánchez-Soto},\ and\
  \citenamefont {Englert}}]{hradil2019quantum}%
  \BibitemOpen
  \bibfield  {author} {\bibinfo {author} {\bibfnamefont {Zdeněk}\ \bibnamefont
  {Hradil}}, \bibinfo {author} {\bibfnamefont {Jaroslav}\ \bibnamefont
  {Řeháček}}, \bibinfo {author} {\bibfnamefont {Luis}\ \bibnamefont
  {Sánchez-Soto}}, \ and\ \bibinfo {author} {\bibfnamefont {Berthold-Georg}\
  \bibnamefont {Englert}},\ }\bibfield  {title} {\enquote {\bibinfo {title}
  {Quantum fisher information with coherence},}\ }\href {\doibase
  10.1364/optica.6.001437} {\bibfield  {journal} {\bibinfo  {journal} {Optica}\
  }\textbf {\bibinfo {volume} {6}},\ \bibinfo {pages} {1437} (\bibinfo {year}
  {2019})}\BibitemShut {NoStop}%
\end{thebibliography}%

\onecolumngrid

\pagebreak

\setcounter{section}{0}
\setcounter{proposition}{0}
\setcounter{theorem}{0}
\setcounter{corollary}{0}

\section*{Appendix}

Here we present the proof for the main results in our manuscript. For simplicity, we reiterate our theorems and propositions.

\section{Proof of Theorem~\ref{th:unitarysm}}
\begin{theorem}[Explicit form]\label{th:unitarysm}
For unitary families, the sub-QFI of Definition~\ref{def:subQFI} can be expressed as 
\begin{align}
    \IC\left(\rho_{\theta}\right)&=4\left(\Tr\left[\rho^2H^2\right]-\Tr\left[\rho H \rho H\right]\right)\label{eq:LquantitySM}\\
    &=-2\Tr\left[\left[\rho,H\right]^2\right]\,.
\end{align}
\end{theorem}
\begin{proof}
First, recall that the super-fidelity defined as 
\begin{align}
    G(\rho,\sigma)=\Tr\left[\rho\sigma\right]+\sqrt{(1-\Tr\left[\rho^2\right])(1-\Tr\left[\sigma^2\right])}\,.
\end{align}
In the case of unitary families, $\rho_{\theta}=W_{\theta}\rho W_{\theta}\ad$ with $W_{\theta}=e^{-i\theta H}$, and hence we have  $\Tr\left[\rho^2\right]=\Tr\left[\rho_{\theta}^2\right]=\Tr\left[\rho_{\theta+\delta}^2\right]$ (the purity is unitarily invariant). Hence, the super fidelity between $\rho_{\theta}$ and $\rho_{\theta+\delta}$ can be simplified as
\begin{align}
    G(\rho_{\theta},\rho_{\theta+\delta}) =\Tr\left[\rho_{\theta}\rho_{\theta+\delta}\right]+1-\Tr\left[\rho^2\right]\,.
\end{align}
The first term can then be expressed as
\begin{align}
    \Tr\left[\rho_{\theta}\rho_{\theta+\delta}\right] &= \Tr\left[(W_{\theta}\rho W^{\ad}_{\theta})( W_{\theta}W_{\delta}\rho W^{\ad}_{\delta}W^{\ad}_{\theta} )\right] \\
    &=\Tr\left[\rho e^{-i\delta H}\rho e^{+i\delta H}\right]\,,
\end{align}
where we have used the definition of the exact and error states and the cyclic property of the trace. 
Then, by employing the series expansion of $e^{-i\theta H}$ when $\abs{\delta}\ll1$
\begin{align}
    e^{\pm i\delta H}=\id\pm i\delta H-\frac{\delta^2}{2}H^2+\OC\left(\delta^3\right)\,,
\end{align}
we can write
\begin{align}
    \Tr\left[\rho_{\theta}\rho_{\theta+\delta}\right] &=\Tr\left[\rho\left(\id - i\delta H-\frac{\delta^2}{2}H^2+\OC\left(\delta^3\right)\right)\rho \left(\id + i\delta H-\frac{\delta^2}{2}H^2+\OC\left(\delta^3\right)\right)\right]\\
    &=\Tr\left[\rho^2 +i \delta \rho^2 H - \frac{\delta^2}{2}\rho^2 H^2 - i \delta \rho H \rho +\delta^2 \rho H \rho H -\frac{\delta^2}{2}\rho H^2 \rho + \OC \left(\delta^3\right)\right]\\
    &= \Tr\left[\rho^2\right] -\delta^2 \left(\Tr\left[\rho^2 H^2\right]-\Tr\left[\rho H \rho H\right] \right) + \OC \left(\delta^3\right)\,.
\end{align}

The super fidelity becomes
\begin{align}
    G(\rho_{\theta},\rho_{\theta+\delta})
    &= \Tr\left[\rho^2\right] -\delta^2 \left(\Tr\left[\rho^2 H^2\right]-\Tr\left[\rho H \rho H\right] \right) + \OC \left(\delta^3\right) + 1 - \Tr\left[ \rho^2\right] \\
    &= 1 -\delta^2 \left(\Tr\left[\rho^2 H^2\right]-\Tr\left[\rho H \rho H\right] \right) + \OC\left(\delta^3\right)\,.
\end{align}
Expand the square root of $G(\rho_{\theta},\rho_{\theta+\delta})$ with respect to $\delta$, we find
\begin{align}
    \begin{split}
    \sqrt{G(\rho_{\theta},\rho_{\theta+\delta})} = 1-\frac{\delta^2}{2}\left(\Tr\left[\rho^2 H^2\right]-\Tr\left[\rho H \rho H\right] \right)+\OC(\delta^3)\,,    
    \end{split}
\end{align}
and rearranging we see from Eq.~\eqref{eq:subQFI} that
\begin{align}
\IC\left(\rho_{\theta}\right)= \lim_{\delta\to 0}8\frac{1-\sqrt{G(\rho_{\theta},\rho_{\theta+\delta})}}{\delta^2} = 4\left(\Tr\left[\rho^2 H^2\right]-\Tr\left[\rho H \rho H\right] \right)\,.
\end{align}

\end{proof}

\appendix

\section{Proof of proposition~\ref{coro:SLDAM}}
\label{app:fullrankQFI}
\begin{proposition}[Non-Hermitian SLD operator]
\label{coro:SLDAM}
For full rank probe states $\rho$, the sub-QFI can be expressed as 
\begin{align}
    \IC\left(\rho_{\theta}\right)=2\Tr\left[\Lambda_{\theta}\ad\Lambda_{\theta}\rho_{\theta}^2\right]\,.
\end{align}
where $\Lambda_{\theta}$ is the  nSLD operator in~\eqref{eq:nSLD}. 
\end{proposition}
\begin{proof}

For a full-rank probe state $\rho$, let us define $\Lambda_{\theta}$ as
\begin{align}
     \Lambda_{\theta} = \left(\partial_{\theta}\rho_{\theta}\right)\rho_{\theta}^{-1}\,, \quad 
     \Lambda_{\theta}\ad = \rho_{\theta}^{-1}\partial_{\theta}\rho_{\theta}\,. 
\label{eq:nSLDSM}
\end{align}

From Eq.~\eqref{eq:nSLDSM} we find that
\begin{align}
    \Tr\left[\Lambda_{\theta}\ad\Lambda_{\theta}\rho_{\theta}^2\right]&=\Tr\left[(\rho_{\theta}^{-1}\partial_{\theta}\rho_{\theta})(\left(\partial_{\theta}\rho_{\theta}\right)\rho_{\theta}^{-1})\rho_{\theta}^2\right]\\
    &=\Tr\left[\left(\partial_{\theta}\rho_{\theta}\right)^2\right]\\
    &=-\Tr\left[\left[\rho_\theta,H\right]^2\right]\label{eq:proof1}\\
    &=-\Tr\left[\left[\rho,H\right]^2\right]\label{eq:proof2}\,.
\end{align}
where in Eq.~\eqref{eq:proof1} we used the fact that since $\rho_{\theta}=W_\theta \rho W\ad_\theta$, then  $\partial_{\theta}\rho_{\theta}=- i H W_\theta \rho W\ad_\theta + i W_\theta \rho W\ad_\theta H= i [\rho_\theta ,H]$. Then, Eq.~\eqref{eq:proof2} is obtained by noting that $\Tr[[\rho_\theta ,H]]=[\rho ,H]$ since $[W_\theta,H]=0$.  Finally, combining the definition of the sub-QFI in Eq.~\eqref{eq:explicit-form-2} with~\eqref{eq:proof2}, we prove Proposition~\ref{coro:SLDAM}.

Here, note that for full-rank states, the QFI can be expressed as  $I(\rho_{\theta})=\Tr\left[\Lambda_{\theta}^2 \rho_{\theta}\right]$ when $[\rho_{\theta},\partial_{\theta}\rho_{\theta}]=0$. In this case, $\Lambda_{\theta}$ becomes the SLD operator~\cite{liu2016quantum}
 \begin{align}
     \Lambda_{\theta}=L_{\theta} = \rho_{\theta}^{-1}\partial_{\theta}\rho_{\theta}=\left(\partial_{\theta}\rho_{\theta}\right)\rho_{\theta}^{-1}\,,
\end{align}
which matches to Eq.~\eqref{eq:standardQFI}.
\end{proof}

\section{Proof of Theorem~\ref{th:maxSM}}

\begin{theorem}[Optimal state]\label{th:maxSM}
For any quantum state $\rho$ and for any Hermitian generator $H$, the sub-QFI of Definition~\ref{def:subQFI} and the  QFI in~\eqref{eq:fidelity} are maximized for the same state preparation $U \rho U\ad$, with the optimal unitary being
\small
\begin{align}\label{eq:UoptSM}
    U_*=\argmax_{U} \IC\left(U\rho U\ad\right)=\argmax_{U}I\left(U\rho U\ad\right)\,.
\end{align}
\normalsize
Here, the maximum is taken over the unitary group of degree $d$, with $d=2^n$.
\end{theorem}

Let us  remark that the proof employed here follows the proof in Ref. \cite{fiderer2019maximal} for determining the optimal mixed state for the QFI. While the main steps in the proof are the same, we nevertheless present here the derivation of the state that maximizes sub-QFI. Moreover, for convenience of the reader, we present additional details in the proof that clarify the steps followed. 

The main idea of the proof is to construct an upper bound, $\mathcal{J}(\rho) \geq \mathcal{I}(\rho)$ , for the sub-QFI $\mathcal{I}(\rho)$, that has simpler behavior than the sub-QFI itself. Note that for simplicity of notation we drop the $\theta$ dependence of $\mathcal{I}$. the If one can find the state, denoted $\rho^*$, that maximizes this upper bound and then show that $\mathcal{J}(\rho^*)=\mathcal{I}(\rho^*)$, one will have shown that $\rho^*$ maximizes the sub-QFI, as desired.

First, recall from Theorem \ref{th:unitary}, that the sub-QFI can be expressed as
\begin{align}
    \mathcal{I}(\rho) =4\left( \Tr[\rho^2H^2]-\Tr[\rho H \rho H]\right)\,. \label{eq:subQFI-TraceForm}
\end{align}
Let $\{\ket{k}\}_{k=1}^d$ denote the eigenbasis for $\rho$, so that 
\begin{align}
    \rho=\sum_{k=1}^d \lambda_k \ket{k}\bra{k}\,.
\end{align}
Then,   the sub-QFI can be expressed as 
\begin{align}
    \mathcal{I}(\rho) = \frac{1}{2} \sum_{i,j=1}^d (\lambda_i-\lambda_j)^2 |h_{ij}|^2,
\end{align}
where $h_{ij}=\langle i|H|j\rangle$ are the matrix elements of $H$ in the eigenbasis of $\rho$, and where we have symmetrized the result to put indices $i,j$ on equal footing. Now, let us define 
\begin{align}
 \lambda_{ij}=
    \begin{cases}
        0 & \text{ if } \lambda_i=\lambda_j=0, \\
        (\lambda_i-\lambda_j)^2 & \text{ else. }
    \end{cases}
\end{align}
As noted in the main text, the sub-QFI and QFI have very similar explicit forms. The only difference is in the form of the coefficients in their explicit representations. It is this observation that makes our problem of finding the optimal state for the sub-QFI amenable to the techniques used in Ref. \cite{fiderer2019maximal}. Now that we have defined the coefficients $\lambda_{ij}$, we prove some essential inequalities for them.
\begin{lemma}\label{lemma1}
Let $\lambda_1 \geq \dots \geq \lambda_d \geq 0.$ Then the following inequalities hold:

\begin{enumerate}
    \item $\lambda_{i,l} \geq \lambda_{i,j}+\lambda_{j,l}$ for $1 \leq i < j< l\leq d$,
    \item $\lambda_{i,l}-\lambda_{i+1,l} \geq \lambda_{i,k} - \lambda_{i+1,k} $ for $1< i+1<k<l \leq d$
    \item $\lambda_{i,l}-\lambda_{i,l-1} \geq \lambda_{j,l}-\lambda_{j,l-1}$ for $1\leq i < j< l-1 < d$
\end{enumerate}
\end{lemma}

\begin{proof}
The three items in the above lemma are special cases of the following inequality:
\begin{align}
    \lambda_{i,l}+\lambda_{j,k}\geq \lambda_{i,k}+\lambda_{j,l} & \quad \text{for } 1 \leq i < j< k< l \leq d\,.\label{eq:lemma11}
\end{align}
To show that~\eqref{eq:lemma11} holds, we note that  $1 \leq i < j< k< l \leq d \implies \lambda_d \leq \lambda_l < \lambda_k < \lambda_j <\lambda_i \leq 1$, and observe
\begin{align}
    \lambda_{i,l}+\lambda_{j,k}- \lambda_{i,k}-\lambda_{j,l}&= (\lambda_i-\lambda_l)^2 + (\lambda_j-\lambda_k)^2 -(\lambda_i-\lambda_k)^2 -(\lambda_j-\lambda_l)^2 \\
    &= 2\left[\lambda_i(\lambda_k-\lambda_l)-\lambda_j(\lambda_k-\lambda_l)\right] \\
    &\geq 0\,.
\end{align}
For the first inequality above, we let $j=k$. For the second and third inequalities, we let $j=i+1$ and $k=l-1$, respectively. Then, the final inequality follows from the fact that $\lambda_i > \lambda_j$. 
\end{proof}

A crucial step in constructing our upper bound is introducing new coefficients defined in terms of the $\lambda_{i,j}$'s as
\begin{align}\label{q-coeff}
    q_{i,j}:=\sum_{k=i}^{j-1} q_{k,k+1}\,,
\end{align}
for $i<j$. These coefficients play a crucial role in the proof of the optimal state due to the following lemma. 
\begin{lemma}\label{lemma2}
For any dimension $d\geq2$ and any ordered, non-negative coefficients $\lambda_1\geq \dots \lambda_d \geq 0$, there exist coefficients $q_{k,k+1} \geq 0$ with $1\leq k \leq d-1$ such that for $1\leq i < j \leq d$:
\begin{align}
    q_{i,j}&=\lambda_{i,j} \quad \text{ if } j=d-i+1, \\
    q_{i,j} &\geq \lambda_{i,j} \quad \text{ else. }
\end{align}
\end{lemma}
Here, although the form of the coefficients in the explicit representation of sub-QFI differ from the QFI, the proof of this proposition is identical to Ref. \cite{fiderer2019maximal}, so we proceed to the proof of Theorem~\ref{th:maxSM}. 

\begin{proof}

As previously mentioned, we will first show that there exists an our upper bound, $\mathcal{J}(\rho) \geq \mathcal{I}(\rho)$, for the sub-QFI $\mathcal{I}(\rho)$ which is maximized by some $\rho^*$. We then show that $\mathcal{J}(\rho^*)=\mathcal{I}(\rho^*)$, which implies that $\rho^*$ maximizes the sub-QFI, as desired. Finally, by inspection, one will be able to see that the state that maximizes the sub-QFI is identical to the state that maximizes the QFI from Ref. \cite{fiderer2019maximal}. Specifically, if we let $\{h_j\}$ and $\{\ket{h_j}\}$ respectively be the set of eigenvalues and associated eigenvectors of $H$ (with $h_j\geq h_{j+1}$). Then, as shown in~\cite{fiderer2019maximal} the state that maximizes the QFI, and hence the sub-QFI, is given by $\rho^{*} =\sum_{j=1}^{d} \lambda_j \dya{\phi_j}$,  where
\begin{equation}\label{eq:optimalSM}
    \ket{\phi_j} =
    \begin{cases}
    \frac{1}{\sqrt{2}}\ket{h_j}+\frac{1}{\sqrt{2}}e^{i\chi}\ket{h_{d-j+1}}&(2j<d+1)\\
    \ket{h_j}&(2j=d+1)\\
   \frac{1}{\sqrt{2}}\ket{h_j}-\frac{1}{\sqrt{2}}e^{i\chi}\ket{h_{d-j+1}}&(2j>d+1)\\
    \end{cases}
\end{equation}
and where $\chi\in\mathbb{R}$ is an arbitrary phase. 

Unitary transformations of the form, $U\rho U^{\dagger}$ leave the eigenvalues of $\rho$ unchanged. This unitary freedom allows one to change between ordered orthonormal bases. As such, optimizing over unitary preparations  is equivalent to optimizing over ordered bases $B\in S$ with
\begin{align}
    S:= \{\ket{e_i}:\braket{e_i}{e_j}=\delta_{i,j} \quad \forall i,j \in \{1,\dots,d\}\}
\end{align}
The only reason one is restricted to \textit{ordered} bases is because the ordered eigenvectors correspond to an ordered set of associated eigenvalues which is crucial for the result. 

Let us denote the optimal ordered basis as $B^*=\{\ket{\phi_i}\}_{i=1}^d$. Then, the maximization problem is solved by 

\begin{align}
    \IC(\rho^*):= \max_{B\in S}\IC (B)
\end{align}
where we have simply expressed the sub-QFI as a function of the ordered basis, $B$. That is,
\begin{align}
    \IC(B) &= 2 \sum_{i,j=1}^d \lambda_{i,j} |h_{i,j}(B)|^2\\
    &= 4 \sum_{i=1}^{d-1} \sum_{j=i+1}^d \lambda_{i,j} |h_{i,j}(B)|^2 
\end{align}
where $h_{i,j}=\bra{e_i}H\ket{e_j}$ are the matrix elements of the generator with respect to the ordered basis $B=\{\ket{e}_i\}_{i=1}^d$.  We remark that $\lambda_{i,j}=\lambda_{j,i}$ and that $|h_{i,j}(B)|^2 =|h_{j,i}(B)|^2 $, as these facts will be used in what follows. 

To obtain the upper bound $ \mathcal{J}(B)$ on the sub-QFI, we replace the $\lambda_{i,j}$ coefficients with the $q_{i,j}$ coefficients from Eq.~\eqref{q-coeff}. That is, defining 
\begin{align}
    \mathcal{J}(B) := 4 \sum_{i=1}^{d-1} \sum_{j=i+1}^d q_{i,j} |h_{i,j}(B)|^2,
\end{align}
we readily find from Lemma~\ref{lemma2} that
\begin{align}
    \IC (B) \leq \mathcal{J}(B)\,.
\end{align}
 Next, using the fact that $q_{i,j}=\sum_{k=i}^{j-1}q_{k,k+1}$ for all $1\leq i < j \leq d$, we can  write
\begin{align}
     \mathcal{J}(B) &= 4 \sum_{i=1}^{d-1} \sum_{j=i+1}^d q_{i,j} |h_{i,j}(B)|^2\\
     &= 4 \sum_{i=1}^{d-1} \sum_{j=i+1}^d \left(\sum_{k=i}^{j-1}q_{k,k+1}\right) |h_{i,j}(B)|^2\\
     &= 4 \sum_{k=1}^{d-1}q_{k,k+1} \left(\sum_{i=1}^k \sum_{j=k+1}^{d} |h_{i,j}(B)|^2\right)\,.
\end{align}
Now, to simplify further, we rewrite the term in the parenthesis as a Hilbert-Schmidt norm. Recalling that the Hilbert-Schmidt norm is defined as $\|A\|_2^2:= \sum_{i=1}^m \sum_{j=1}^n |A_{i,j}|^2$, one we can identify the parenthetical term as the Hilbert-Schmidt norm of the subblock of $H$ spanning from the $1$st to $k$-th row and from the $(k+1)$-th to the $d$-th column. With this insight, one can write
\begin{align}
    \mathcal{J}(B) = 4 \sum_{k=1}^{d-1} q_{k,k+1} \| h(B,k)\|_2^2,
\end{align}
where as in Ref. \cite{fiderer2019maximal} we denote the aforementioned subblock as $h(B,k)$. 

Expressing the upper bound in terms of Hilbert-Schmidt norms allows us to use the Bloomfield-Watson inequality, which bounds the norm of off-diagonal blocks of matrices. The relevant form of the inequality is 
\begin{align}
     \| h(B,k)\|_2^2 \leq \frac{1}{4} \sum_{i=1}^{m(k)} (h_i-h_{d-i+1})^2, \label{eq:BW-inequality}
\end{align}
where $m(k) = \min{\{k,d-k\}}$. If the conjectured optimal basis $\{\ket{\phi_i}\}_{i=1}^d$ of Eq.~\eqref{eq:optimalSM} is optimal, then the left-hand side of the Bloomfield-Watson inequality is saturated. Explicitly, we can write
\begin{align}
     \| h(B^*,k)\|_2^2&=\sum_{i=1}^k \sum_{j=k+1}^{d} |h_{i,j}(B^*)|^2 \\
     &= \sum_{i=1}^k \sum_{j=k+1}^{d} |\bra{\phi_i}H\ket{\phi_j}|^2\,.
\end{align}
Then, from Eq. \eqref{eq:optimal} we see that the form of the optimal eigenvectors is different when $2k<d+1$, $2k=d+1$, or $2k>d+1$. When $i=1$, we see that $\ket{\phi_1}$ contains both $\ket{h_1}$ and $\ket{h_d}$. Then for $\ket{\phi_2}$ we will have both $\ket{h_2}$ and $\ket{h_{d-1}}$, and so on. Hence, computing the overlaps $\bra{\phi_i}H\ket{\phi_j}$ will yield delta functions of the form $\delta_{i,j}$ (prohibited from the fact that the sum is over off-diagonal terms) and $\delta_{i,d-j+1}$. Then, we are left with
\begin{align}
     \| h(B^*,k)\|_2^2&= \sum_{i=1}^k \sum_{j=k+1}^{d} |\bra{\phi_i}H\ket{\phi_j}|^2\\
     &= \sum_{i=1}^k \sum_{j=k+1}^{d} \left(\delta_{i,d-j+1}\frac{(h_i-h_{d-i+1})}{2}\right)^2\\
     &=\frac{1}{4}\sum_{i=1}^{m(k)} (h_i -h_{d-i+1})^2\,,\label{eq:final-sum}
\end{align}
where the final sum over $i$ ranges from $1$ to $m(k)= \min{\{k,d-k\}}$. This is to avoid double counting in the case where $k>d/2$. Comparing Eq.~\eqref{eq:final-sum} to Eq.~\eqref{eq:BW-inequality}, we see that the Bloomfield-Watson inequality is in fact saturated when $B=B^*$. The Hilbert-Schmidt norm of the off-diagonal blocks of the generator are thus maximized for $B=B^*$ meaning $\|h(B,k)\|_2^2\leq \|h(B^*,k)\|_2^2$ which implies
\begin{align}
    \mathcal{J}(B) \leq \mathcal{J}(B^*) \quad \forall B\in S\,.
\end{align}

Now that we have established that $B^*$ maximizes the upper bound on the QFI, it remains to be shown that $\mathcal{J}(B^*)=\IC(B^*)$. From the definition of $\mathcal{J}(B)$, we have
\begin{align}
    \mathcal{J}(B^*) &= 4 \sum_{i=1}^{d-1} \sum_{j=i+1}^d q_{i,j} |h_{i,j}(B^*)|^2 \\
    &= 4 \sum_{i=1}^{d-1} \sum_{j=i+1}^d q_{i,j} \left(\delta_{i,d-j+1}\frac{h_i-h_{d-i+1}}{2}\right)^2\\
    &=\frac{1}{2}\sum_{i=1}^{d} \lambda_{i,d-i+1} \left(h_i-h_{d-i+1}\right)^2 \\
    \mathcal{J}(B^*)&=\IC(B^*)\,,
\end{align}
where we have used Lemma~\ref{lemma1} to assert $q_{i,d-i+1}=\lambda_{i,d-i+1}$ and where we revert to a sum from $i=1$ to $i=d$ which requires a factor of $1/2$ to avoid double counting. We have now established that the Sub-QFI and QFI are maximized by the same state. This is equivalent to saying they are maximized by the same unitary state preparation. Thus, the Theorem \ref{th:max} from the main text has been proven.
\end{proof}

\section{Proof of Proposition~\ref{prop:geo}}
\begin{proposition}[Geometrical interpretation] 
The sub-QFI can be expressed as
\begin{equation}
    \IC\left(\rho_{\theta}\right)= \partial^2_{\theta}D_{HS}(\rho,\rho_{\theta})\Bigr\rvert_{\theta = 0}\,,
\end{equation}
where $D_{HS}(A,B)=\Tr[(A-B)^2]$ is the Hilbert-Schmidt distance. 
\end{proposition}

\begin{proof}
\begin{align}
    \partial^2_{\theta}D_{HS}(\rho,\rho_{\theta})&=\partial^2_{\theta} \left(\Tr\left[(\rho-\rho_{\theta})^2\right]\right), \quad \quad (\text{ definition of } D_{HS}(\rho,\rho_{\theta}))\\ 
    &= \partial^2_{\theta} \left(\Tr\left[\rho^2 -\rho \rho_{\theta} -\rho_{\theta} \rho +\rho_{\theta}^2\right]\right), \\
    &= 2\partial^2_{\theta} \left(\Tr\left[\rho^2\right] - \Tr\left[\rho \rho_{\theta} \right]\right), \quad \quad (\text{ linearity and cyclicity of trace})\\
    &= -2 \Tr\left[\rho~ \partial^2_{\theta}\rho_{\theta} \right]\,. \quad \quad (\partial_{\theta}\Tr[\rho^2]=0) \label{eq:c5}
\end{align}
The first derivative of the exact state with respect to $\theta$ can be shown to be $\partial_{\theta} \rho_\theta= i e^{-i\theta H} [\rho,H]e^{+i\theta H}$. The second derivative is then
\begin{align}
    \partial_{\theta}^2 \rho_\theta &= \partial_{\theta} \left[i e^{-i\theta H} [\rho,H]e^{+i\theta H} \right],\\
    &= He^{-i\theta H}[\rho,H]e^{i\theta H} - e^{-i\theta H}[\rho,H]He^{i\theta H}.
\end{align}
So, substituting back into Eq. \eqref{eq:c5}, we have
\begin{align}
    \partial^2_{\theta} D_{HS}(\rho,\rho_{\theta}) &= -2\Tr\left[\rho( He^{-i\theta H}[\rho,H]e^{i\theta H} - e^{-i\theta H}[\rho,H]He^{i\theta H})\right],\\
    &= -2 \Tr\left[\rho H \rho_{\theta} H - \rho H H \rho_{\theta} -H\rho \rho_{\theta}H + H\rho H \rho_{\theta}\right], 
\end{align}
where we have used $[H,e^{\pm i \theta H}]=0,$ and the linearity and cyclicity of trace. Now, letting $\theta \rightarrow 0$ we arrive at the final expression
\begin{align}
     \lim_{\theta\to0}\partial^2_{\theta} D_{HS}(\rho,\rho_{\theta}) &=-2\Tr\left[ [\rho,H]^2\right].
\end{align}
Comparing to Eq. \eqref{eq:explicit-form-2}, we see that 
\begin{align}
    \IC ( \rho_{\theta}) = \partial^2_{\theta} D_{HS}( \rho_{\theta})\Bigr\rvert_{\theta = 0}\,.
\end{align}
\end{proof}

\end{document}